\documentclass[runningheads,9pt]{llncs}

\usepackage{amsmath}
\usepackage{amssymb}
\usepackage{graphicx}
\usepackage{fink}
\usepackage[table]{xcolor}
\usepackage[font=footnotesize]{caption}
\usepackage{stmaryrd}
\usepackage{multirow}
\usepackage{mathpartir}
\usepackage{subfig}
\usepackage{caption}
\usepackage{float}
\usepackage{hyperref}
\usepackage{algorithm,algorithmic}

\newcommand{\dist}[1]{\text{\em Dist}(#1)}

\renewcommand{\L}{\mathcal{L}}

\newcommand{\dirac}[1]{\delta_{#1}}
\newcommand{\ie}{{\em i.e. }}

\newcommand{\preorder}{\preceq}
\newcommand{\dpreorder}{\sqsubseteq}
\newcommand{\trans}[1]{\xrightarrow{#1}}

\newcommand{\supp}[1]{\text{{\em Supp}}(#1)}
\newcommand{\project}[1]{\upharpoonright_{#1}}

\newcommand{\rat}{\mathbb{Q}}
\newcommand{\product}{\otimes}

\newcommand{\old}{\text{{\em old}}}
\newcommand{\java}{Java$^{\text{TM}}$ }

\newtheorem{notation}{Notation}

\title{Assume-Guarantee Abstraction Refinement\\ for Probabilistic Systems
\thanks{
This research was sponsored by DARPA META II, GSRC, NSF, 
SRC, GM, ONR under contracts FA8650-10C-7079,
1041377 (Princeton University), CNS0926181/CNS0931985, 2005TJ1366, 
GMCMUCRLNV301, N000141010188, respectively, and the CMU-Portugal Program.
This is originally published by Springer-Verlag as part of the proceedings of
CAV 2012 and is available at www.springerlink.com. URL for the publication :
\url{http://dx.doi.org/10.1007/978-3-642-31424-7_25}.
}}
\titlerunning{AGAR for Probabilistic Systems}

\author{Anvesh Komuravelli\inst{1} \and Corina S. P\u{a}s\u{a}reanu\inst{2} \and
Edmund M. Clarke\inst{1}}
\authorrunning{A.~Komuravelli et al.}
\institute{Computer Science Department, Carnegie Mellon University, Pittsburgh,
PA, USA
\and Carnegie Mellon Silicon Valley, NASA Ames, Moffett Field, CA, USA\\
}

\begin{document}
\maketitle

\begin{abstract}
We describe an automated technique for assume-guarantee style
checking of strong simulation between a system and a specification, both
expressed as non-deterministic Labeled Probabilistic Transition
Systems (LPTSes). We first characterize counterexamples to strong
simulation as {\em stochastic} trees and show that simpler structures
are insufficient. Then, we use these trees in an abstraction
refinement algorithm that computes the assumptions for
assume-guarantee reasoning as conservative LPTS abstractions of some
of the system components. The abstractions are automatically refined
based on tree counterexamples obtained from failed simulation checks
with the remaining components.  
%The technique further applies to checking safe-PCTL properties (not considered here). 
We have implemented the algorithms for counterexample generation and
assume-guarantee abstraction refinement 
and report encouraging results.
%outperform monolithic?
\end{abstract}

\section{Introduction}

Probabilistic systems are increasingly used for the formal modeling and analysis
of a wide variety of systems ranging from randomized communication and security protocols to nanoscale computers and biological
processes.
%%There is an increasing interest in the verification of probabilistic
%%systems due to their applications in randomized distributed protocols,
%%computer security and systems biology, to name a few.
Probabilistic model checking is an
automatic technique for the verification of such systems against
formal specifications \cite{BK_book}. However, as in the classical
non-probabilistic case \cite{CGP_book}, it %the technique %probabilistic model checking
suffers from the {\em state explosion} problem, where the state space
of a concurrent system grows exponentially in the number of its
components.

%In the non-probabilistic case, 
%Abstraction~\cite{CGL_toplas94,CGJ+_cav00} and compositional verification
%\cite{CLM_lics89} are two extensively studied approaches to combat
%this problem. 
%Recently, these approaches have also been investigated
%in the context of probabilistic systems
%\cite{HWZ_cav08,CV_tocl10,KNP+_tacas10,FKP_fase11}.

Assume-guarantee style compositional techniques~\cite{Pnueli_lmcs85}
address this problem by decomposing the verification of a system into
that of its  smaller components and composing back the
results, without verifying the whole system directly. When checking
individual components, the method uses {\em assumptions} about the
components' environments and then, discharges them on the rest of the
system. For a system of two components, such reasoning is captured by
the following simple assume-guarantee rule.

\vspace{-0.1in}
  \begin{mathpar}
  \inferrule*[right=$(${\sc ASym}$)$]
            {1: L_1 \parallel A \preorder P \\ 2: L_2 \preorder A}
            {L_1 \parallel L_2 \preorder P}
  \end{mathpar}
\vspace{-0.2in}

Here $L_1$ and $L_2$ are system components, $P$ is a specification
to be satisfied by the composite system and $A$ is
an assumption on $L_1$'s environment, to be discharged on $L_2$.
Several other such
rules have been proposed, some of them involving symmetric~\cite{PGB+_fmsd08} or
circular~\cite{AHJ_concur01,PGB+_fmsd08,KNP+_tacas10} reasoning. 
Despite its simplicity, rule {\sc ASym} has been proven the most effective in practice and
studied extensively~\cite{PGB+_fmsd08,CCS+_cav05,FKP_fase11}, mostly in the context of
non-probabilistic reasoning.

We consider here the {\em automated} assume-guarantee style compositional
verification of {\em Labeled Probabilistic Transition Systems}
(LPTSes), whose transitions have both probabilistic and non-deterministic
behavior. The verification is performed using the rule {\sc ASym} where
$L_1$, $L_2$, $A$ and $P$ are LPTSes and the conformance relation $\preorder$ is instantiated with
{\em strong simulation} \cite{SL_nordic95}.
We chose strong simulation for the following reasons.  Strong
simulation is a decidable, well studied relation between
specifications and implementations, both for non-probabilistic
\cite{Milner_tr71} and probabilistic \cite{SL_nordic95} systems. A
method to help scale such a check is of a natural interest.
Furthermore, rule {\sc ASym} is both sound and complete for this relation.
Completeness is obtained trivially by replacing $A$ with $L_2$ but is essential for
full automation (see Section~\ref{sec:agar}). One can argue that strong
simulation is too fine a relation to yield suitably small
assumptions. However, previous success in using strong simulation in
non-probabilistic compositional verification \cite{Chaki_thesis05} motivated us
to consider it in a probabilistic setting as well.
And we shall see that indeed we can obtain small assumptions for the examples we consider
while achieving savings in time and memory (see Section \ref{sec:results}).

%Here $L_1$, $L_2$, $A$ and $P$ denote
%labeled probabilistic systems, $\parallel$ denotes parallel
%composition and $\preorder$ denotes (strong) probabilistic simulation
%conformance (see Section ...). 
%%The rule is both sound and complete; 
%Completeness follows simply from the fact that one can replace $A$
%with $L_2$ in the rule premises.  Note however that for the rule to be
%useful in practice, assumption $A$ needs to be much smaller than than
%$L_2$ but still reflect $L_2$'s behavior in such a way that the second
%premise is satisfied. Coming up with such assumptions manually is
%highly non-trivial.  Following previous work on assume-guarantee
%abstraction refinement (AGAR)~\cite{AGARCAV08} that was done in the
%context of non-probabilistic verification, we propose here an approach 
%for building assumptions {\em automatically} in the
%probabilistic setting, using the assume-guarantee rule above.

The main challenge in automating assume-guarantee reasoning is to 
come up with such small assumptions satisfying the premises. 
In the non-probabilistic case, solutions to this problem
have been proposed which use either automata learning techniques
\cite{PGB+_fmsd08,CCS+_cav05} or abstraction refinement
\cite{BPG_cav08} and several improvements and optimizations followed.
For probabilistic systems, techniques using automata learning 
%from trace counterexamples 
have been proposed. They target {\em probabilistic reachability} checking and are not guaranteed to terminate due to incompleteness of the assume-guarantee rules \cite{FKP_fase11} or to the undecidability of the conformance relation and learning algorithms used
\cite{FHK+_atva11}.

%approach or undecidability of the conformance \cite{FKP_fase11,FHK+_atva11}.
%in the context of assume-guarantee style compositional checking of a particular kind of properties, namely
%{\em probabilistic reachability} properties \cite{FKP_fase11,FHK+_atva11}, but the
%proposed learning-based methods are not guaranteed to terminate, due
%either to the incompleteness of the assume-guarantee rules
%\cite{FKP_fase11} or to the undecidability of the learning algorithms
%\cite{FHK+_atva11}.
%On the other hand, our approach is based on abstraction refinement with a
%guarantee of termination (see Section \ref{sec:agar}).

In this paper we propose a complete, fully automatic framework for the
compositional verification of LPTSes with respect to simulation
conformance. One fundamental ingredient of the framework is the use of
{\em counterexamples} (from failed simulation checks) to iteratively
refine inferred assumptions.
Counterexamples are also extremely useful in general to help with debugging of discovered errors.
%and have also been used in the work described above to similarly perform assumption refinement.
However, to the best of our
knowledge, the notion of a counterexample has not been previously
formalized for strong simulation between probabilistic systems. As our
first contribution we give a characterization of
counterexamples to strong simulation  as {\em stochastic} trees
and an algorithm to compute them; we also show that simpler structures
are insufficient in general (Section~\ref{sec:cex}).

We then propose an assume-guarantee abstraction-refinement (AGAR)
algorithm (Section \ref{sec:agar}) to automatically build the assumptions used in compositional
reasoning. % for probabilistic simulation conformance.   
The algorithm follows previous work~\cite{BPG_cav08} which, however, was done in a
non-probabilistic, trace-based setting. % and hence it is very different from ours.
In our approach, $A$ is maintained as a {\em conservative abstraction} of $L_2$,
\ie an LPTS that simulates $L_2$ (hence, 
premise 2 holds by construction), and is iteratively refined based on
tree counterexamples obtained from checking premise 1.
%Since the rule is complete,
The iterative process 
%refinement converges to $L_2$ in the worst case and, since {\sc ASym} is trivially complete, AGAR 
is guaranteed to terminate, with the number of iterations bounded by 
the number of states in $L_2$.
%Note that, a complete rule is necessary for termination, as also mentioned
%previously.
%showing that either the conclusion holds or fails, in which case a
%counterexample is reported. 
When $L_2$ itself is composed of multiple components, the second
premise ($L_2 \preorder A$) is viewed as a new compositional check,
generalizing the approach to $n \ge 2$ components.  AGAR can be further
applied to the case where the specification $P$ is instantiated with a
formula of a logic preserved by strong simulation, such as {\em safe}-pCTL.

We have implemented the algorithms for
counterexample generation and for AGAR
%using PRISM's \cite{KNP_cav11} explicit-state engine
using \java and Yices \cite{yices} and show experimentally that AGAR
can achieve significantly better performance than non-compositional
verification.

%All lemmas and theorems are proved in a longer version of this paper \cite{KPC_arxiv}.

\vspace{0.1in}
\noindent{\bf Other Related Work.}
Counterexamples to strong simulation have been characterized
before as tree-shaped structures for the case of non-probabilistic
systems~\cite{Chaki_thesis05} which we generalize to stochastic trees in Section \ref{sec:cex} for the probabilistic
case.  Tree counterexamples have also been used in the context of a compositional framework that uses
rule {\sc ASym} for checking strong simulation  in
the non-probabilistic case \cite{CCS+_cav05} and employs tree-automata
learning to build deterministic assumptions.

%AGAR is a variant of the well-known CounterExample Guided Abstraction Refinement (CEGAR) approach
%\cite{CGJ+_cav00} which was adapted to probabilistic systems in~\cite{HWZ_cav08,CV_tocl10}.  
%The CEGAR approach we describe in Section \ref{sec:cegar} is an adaptation of the latter.
%Both ~\cite{HWZ_cav08,CV_tocl10} consider abstraction refinement in a
%monolithic, non-compositional setting and do not use 
%counterexamples from checking one component to refine the abstraction of another component, 
%as we do in AGAR. CEGAR has been used before in compositional reasoning in a
%non-probabilistic setting~\cite{Chaki_thesis05}. In that work conservative abstractions of
%individual components are constructed and then all the abstractions are composed and checked,
%which can be expensive if there are many components; that work does not use assume-guarantee
%reasoning. In contrast AGAR does not compose all the abstractions together --
%see Section~\ref{sec:agar} for a description of how a system of $n \ge 2$ components is checked
%through a recursive application of assume-guarantee reasoning.

AGAR is a variant of the well-known CounterExample Guided Abstraction Refinement (CEGAR) approach
\cite{CGJ+_cav00}. CEGAR has been adapted to probabilistic systems, in the
context of probabilistic reachability \cite{HWZ_cav08} and safe-pCTL \cite{CV_tocl10}. The
CEGAR approach we describe in Section \ref{sec:cegar} is an adaptation of the
latter. Both these works consider abstraction refinement in a monolithic, non-compositional
setting. On the other hand, AGAR uses counterexamples from checking one
component to refine the abstraction of another component.

\section{Preliminaries}
\label{sec:prelims}
%\vspace{-0.1in}
%We introduce some notation, definitions and previously known results which we
%will use for the rest of the paper.
\noindent{\bf Labeled Probabilistic Transition Systems.}
Let $S$ be a non-empty set. $\dist{S}$ is defined
to be the set of discrete probability distributions %(the probabilities sum to $1$)
over $S$. We assume that all the
probabilities specified explicitly in a distribution are rationals in $[0,1]$;
there is no unique
representation for all real numbers on a computer and floating-point numbers are
essentially rationals.
For $s \in S$, $\dirac{s}$ is the Dirac distribution on $s$, \ie
$\dirac{s}(s) = 1$ and $\dirac{s}(t) = 0$ for all $t \neq s$. For
$\mu \in \dist{S}$, the {\em support} of $\mu$,
denoted $\supp{\mu}$, is defined to be the set $\{s \in S |
\mu(s) > 0\}$ and for $T \subseteq S$, $\mu(T)$ stands for
$\sum_{s \in T} \mu(s)$.
The models we consider, defined below, have both probabilistic and
non-deterministic behavior. Thus, there can be a non-deterministic choice
between two probability distributions, even for the same action. Such modeling
is mainly used for underspecification and moreover, the abstractions we consider (see Definition \ref{def:quotient_lpts})
naturally have this non-determinism. As we see below, the theory described
does not become any simpler by disallowing non-deterministic choice for a given
action (Lemmas \ref{lem:cex_fully_prob_insuff} and
\ref{lem:cex_reactive_insuff}).
%Below, we define a labeled transition
%system with both probabilistic and non-deterministic transitions.
%We assume that $\supp{\mu}$ is finite for all such $\mu$.

%Let $\mu \in \dist{X} \cup \dist{X}$. Viewing $\mu$ as a partial function from $X$ to $\nnreal$, we
%define $\dist{\mu}$ as the set of all partial functions {\em contained
%within} $\mu$ and hence $\dist{\mu}$ is a set of sub-distributions.
% Note that this also
%implies $\dist{\mu}$ is finite.

\begin{definition}[LPTS]
A {\em Labeled Probabilistic Transition System} $($LPTS$)$ is a tuple $\langle
S, s^0, \alpha, \tau \rangle$ where $S$ is a set of states, $s^0
\in S$ is a distinguished start state, $\alpha$ is a set of actions and $\tau
\subseteq S \times \alpha \times \dist{S}$ is a probabilistic transition
relation.
For $s \in S$, $a \in \alpha$ and $\mu \in \dist{S}$, we denote
$(s,a,\mu) \in \tau$ by $s \trans{a} \mu$ and say that $s$ has a
{\em transition} on $a$ to $\mu$.

An LPTS is called {\em reactive} if $\tau$ is a partial function from $S \times
\alpha$ to $\dist{S}$ $(${\em \ie}at most one transition on a given action from a given
state$)$ and {\em fully-probabilistic} if $\tau$ is a partial function from $S$
to $\alpha \times \dist{S}$ $(${\em \ie}at most one transition from a given state$)$.
\end{definition}

\begin{figure}[t]
\centering
\includegraphics[scale=1.8]{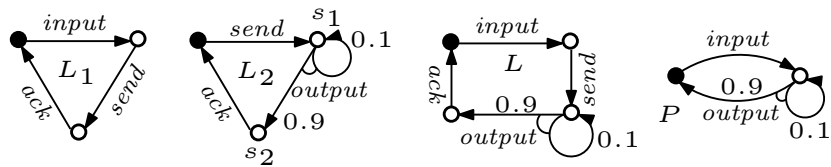}
\caption{Four reactive and fully-probabilistic LPTSes.}
\vspace{-0.2in}
\label{fig:io}
\end{figure}

%\begin{figure}
%\centering
%\includegraphics[scale=1.8]{figures/lpts.eps}
%\caption{Three LPTSes. $L_2$ is reactive and $L_3$ is fully-probabilistic. $L_1$
%is neither.}
%\label{fig:lpts}
%\end{figure}
%{\bf Def 5 comes much later. Also, not very important. Remove?}
%One can generalize the transition relation by allowing {\em
%sub}-distributions where the sum of the probabilities, say {\em sum}, can be less than $1$ with the
%intended meaning that the LPTS deadlocks in the current state with probability
%$1 - \text{{\em sum}}$. By changing Definition \ref{def:strong_simulation}
%appropriately, all the results presented in this paper are
%directly applicable.

Figure \ref{fig:io} illustrates LPTSes. Throughout this paper, we use
filled circles to denote start states in the pictorial representations of
LTPSes. For the distribution $\mu = \{(s_1,0.1),(s_2,0.9)\}$,
$L_2$ in the figure has the transition $s_1 \trans{output} \mu$. All the LPTSes
in the figure are {\em reactive} as no state has more than one transition on a given
action.  They are also {\em fully-probabilistic} as no state has more
than one transition. In the literature, an LPTS is also called a {\em
  simple probabilistic automaton} \cite{SL_nordic95}. Similarly, a
reactive (fully-probabilistic) LPTS is also called a (Labeled) {\em Markov Decision Process}
({\em Markov Chain}). Also, note that an LPTS with all the distributions
restricted to Dirac distributions is the classical (non-probabilistic)
{\em Labeled Transition System} (LTS); thus a {\em reactive} LTS
corresponds to the standard notion of a {\em deterministic}
LTS. For example, $L_1$ in Figure \ref{fig:io} is a reactive (or deterministic)
LTS. We only consider finite state, finite
alphabet and finitely branching (\ie finitely many transitions from
any state) LPTSes.

We are also interested in LPTSes with a tree structure, \ie the start state is
not in the support of any distribution and every other state is in the support
of exactly one distribution. We call such LPTSes {\em stochastic trees} or simply,
{\em trees}.

We use $\langle S_i, s^0_i, \alpha_i, \tau_i
\rangle$ for an LPTS $L_i$ and $\langle S_L, s^0_L, \alpha_L, \tau_L \rangle$
for an LPTS $L$. The following notation is used in Section \ref{sec:agar}.

%We use the variables $L$, $L_1$, $L_2$, $\dots$ (among others) to stand for LPTSes and also use these
%subscripts for the various components of the tuples of the respective LPTSes.
%For example $L_1$ is the tuple $\langle S_1, s^0_1, \alpha_1, \tau_1 \rangle$.
%When necessary, we subscript the components of the tuples by the name of the
%LPTS, for e.g., $\langle S_L, s^0_L, \alpha_L, \tau_L \rangle$ for $L$.
%Further, for $s \in S$, we let $(L,s)$ stand for $\langle S, s, \alpha, \tau
%\rangle$, \ie $L$ with $s$ as the start state.

%For an LPTS $L$, we define $\dist{L}$ to be the union of
%$\dist{\mu}$ for all the sub-distributions $\mu$ appearing in $L$.

%Further, an LTS is a PA where the only allowed sub-distributions are Dirac
%distributions and for convenience, each such Dirac distribution is replaced by the
%single state in the support. Thus, an LTS is a PA with $\tau \subseteq S \times
%\alpha \times S$.

%We use the variables $L_1$, $L_2$, \dots to refer to LTSes, $M_1$, $M_2$, \dots to refer to
%MDPs and $P_1$, $P_2$, \dots to refer to PAs. And for two LTSes $L_1$ and $L_2$,
%we refer to their start states by $s^0_1$ and $s^0_2$, respectively. Similarly,
%we refer to the start states of two MDPs or PAs.

\begin{notation}
For an LPTS $L$ and an alphabet $\alpha$ with $\alpha_L \subseteq \alpha$,
$L^\alpha$ stands for the LPTS $\langle S_L, s^0_L, \alpha, \tau_L \rangle$.
\end{notation}

%Below, we define LPTSes with a special tree structure.

%\begin{definition}[Tree]
%\label{def:tree}
%Given an alphabet $\alpha$, a {\em tree} is any string generated by the
%following grammar,
%\[
%  T := \lambda | a \cdot (p_1 : T \oplus p_2 : T \oplus \dots \oplus p_n : T) | T + T
%\]
%where $\lambda$ denotes the empty tree, $a \in \alpha$, $n \in \posnat$, $p_i
%\in \rat \cap [0,1]$ for $1 \le i \le n$ and $\sum_{i=0}^n p_i = 1$.
%\end{definition}
%
%For example, $a \cdot (0.5 : (b \cdot (1 : \lambda) + c \cdot (1 : \lambda))
%\oplus 0.5 : \lambda)$ is a tree.
%A tree is naturally seen as a {\em tree-shaped} LPTS with the
%root as the initial state and whose transitions are
%given by the second production rule and branching is given by the third
%production rule of the grammar in Definition \ref{def:tree}. For example, $C$ in
%Figure \ref{fig:cex_eg} is the LPTS corresponding to the above tree. In the sequel, we
%use {\em tree} to refer to this corresponding tree-shaped LPTS.

%If the sub-distributions are restricted to Dirac distributions, we obtain
%(non-stochastic) trees which can naturally be seen as {\em tree-shaped} LTSes.

Let $L_1$ and $L_2$ be two LPTSes and $\mu_1 \in \dist{S_1}$,
$\mu_2 \in \dist{S_2}$.
%We first define the product of two distributions and then the
%parallel composition of two LPTSes.
%in the CSP style \cite{Hoare_cacm78}.

\begin{definition}[Product \cite{SL_nordic95}]
\label{def:product}
The product of $\mu_1$ and $\mu_2$, denoted $\mu_1 \product \mu_2$, is a
distribution in $\dist{S_1 \times S_2}$, such that $\mu_1 \product
\mu_2 : (s_1,s_2) \mapsto \mu_1(s_1) \cdot \mu_2(s_2)$.
\end{definition}

\begin{definition}[Composition \cite{SL_nordic95}]
\label{def:composition}
The parallel composition of $L_1$ and $L_2$, denoted $L_1 \parallel L_2$, is
defined as the LPTS $\langle S_1 \times S_2, (s^0_1,s^0_2),
\alpha_1 \cup \alpha_2, \tau \rangle$ where $((s_1,s_2),a,\mu) \in \tau$ iff
  \begin{enumerate}
  \item $s_1 \trans{a} \mu_1$, $s_2 \trans{a} \mu_2$ and $\mu =
\mu_1 \product \mu_2$, or
  \item $s_1 \trans{a} \mu_1$, $a \not\in \alpha_2$ and $\mu = \mu_1
\product \dirac{s_2}$, or
  \item $a \not\in \alpha_1$, $s_2 \trans{a} \mu_2$ and $\mu = \dirac{s_1}
\product \mu_2$.
  \end{enumerate}
\end{definition}

For example, in Figure \ref{fig:io}, $L$ is the composition of $L_1$ and $L_2$.

%See Figure \ref{fig:comp} for an example of composition.
%Next, we define a relation between two distributions and then define strong simulation between
%$L_1$ and $L_2$. Let $R \subseteq S_1 \times S_2$. We define $S^d_i$ to be the set obtained by adding a new dummy
%member $d_i$ to $S_i$, for $i = 1$ or $2$ and $R^d \subseteq S^d_1 \times S^d_2$
%to be $R \cup (\{d_1\} \times S^d_2)$. Intuitively, as the probabilities in $\mu_1$ and $\mu_2$
%need not sum to $1$, we introduce dummy states for the remaining probability and
%lift $R$ to $R^d$ as defined above.
%
%\begin{definition}[\cite{SL_nordic95,Zhang_thesis08}]
%\label{def:weight_function_based}
%$\mu_1 \dpreorder_R \mu_2$ iff there is a {\em weight} function $w :
%S^d_1 \times S^d_2 \rightarrow \rat \cap [0,1]$ such that
%\begin{enumerate}
%  \item $\mu_1(s_1) = \sum_{s_2 \in S^d_2} w(s_1, s_2)$ for all $s_1 \in S^d_1$,
%  \item $\mu_2(s_2) = \sum_{s_1 \in S^d_1} w(s_1, s_2)$ for all $s_2 \in S^d_2$,
%  \item $w(s_1, s_2) > 0$ implies $s_1 R^d s_2$ for all $s_1 \in S^d_1$, $s_2
%\in S^d_2$.
%  \end{enumerate}
%\end{definition}

%See Figure \ref{fig:comp} for an example of composition.

\vspace{0.1in}
\noindent{\bf Strong Simulation.}
For two LTSes, a pair of states
belonging to a strong simulation relation depends on whether certain
other pairs of successor states also belong to the relation
\cite{Milner_tr71}. For LPTSes, one has successor {\em distributions}
instead of successor states; a pair of states belonging to a strong
simulation relation $R$ should now depend on whether certain other
pairs in the {\em supports} of the successor distributions also
belong to $R$. Therefore we define a binary relation on
distributions, $\dpreorder_R$, which depends on the relation $R$
between states.
Intuitively, two distributions can be related if we can pair the states in their
support sets, the pairs contained in $R$, {\em matching all} the probabilities under the
distributions. 

Consider an example with $s R t$ and the transitions $s \trans{a}
\mu_1$ and $t \trans{a} \mu_2$ with $\mu_1$ and $\mu_2$ as in Figure
\ref{fig:dist_rel}(a). In this case, one easy way to match the probabilities is to pair $s_1$ with $t_1$ and
$s_2$ with $t_2$. This is sufficient if $s_1 R t_1$ and $s_2 R t_2$ also hold,
in which case, we say that $\mu_1 \dpreorder_R \mu_2$.
However, such a direct matching may not be possible in general, as is the case in
Figure \ref{fig:dist_rel}(b). One can
still obtain a matching by {\em splitting} the probabilities under the distributions in such a way
that one can then directly match the probabilities as in Figure
\ref{fig:dist_rel}(a). Now, if $s_1 R t_1$, $s_1 R
t_2$, $s_2 R t_2$ and $s_2 R t_3$ also hold, we say that $\mu_1 \dpreorder_R \mu_2$.
Note that there can be more than one possible splitting.
This is the central idea behind the following definition where the splitting is achieved by a {\em
weight function}. 
%For the rest of the section, 
Let $R \subseteq S_1 \times S_2$.
%be arbitrary.

\begin{figure}[t]
\centering
  \subfloat[] {
  \includegraphics[scale=1.5]{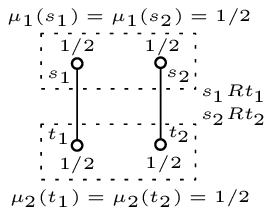}
  }
\hspace{0.5in}
  \subfloat[] {
  \includegraphics[scale=1.5]{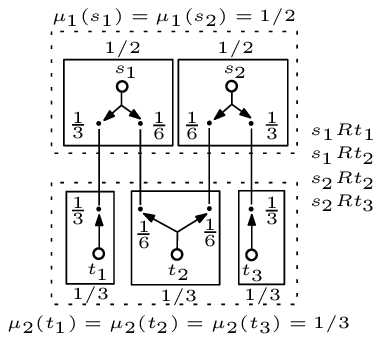}
  }
\caption{Explaining $\mu_1 \dpreorder_R \mu_2$ by means of splitting (indicated
by arrows) and matching (indicated by solid lines) the probabilities.}
\vspace{-0.2in}
\label{fig:dist_rel}
\end{figure}

%\begin{figure}
%\begin{minipage}[c]{0.5\linewidth}
%\centering
%\includegraphics[scale=1.5]{figures/dist_rel_simple.eps}
%\vspace{0.1cm}
%{\scriptsize (a)}
%\end{minipage}
%\begin{minipage}[c]{0.5\linewidth}
%\centering
%\includegraphics[scale=1.5]{figures/dist_rel_general.eps}
%{\scriptsize (b)}
%\end{minipage}
%\caption{Explaining $\mu_1 \dpreorder_R \mu_2$ by means of splitting (indicated
%by arrows) and matching (indicated by solid lines) the probabilities.}
%\label{fig:dist_rel}
%\end{figure}

%Next, we define when two distributions are related and define strong
%simulation between $L_1$ and $L_2$. We also mention several useful
%lemmas.  Let $R \subseteq S_1 \times S_2$.

\begin{definition}[\cite{SL_nordic95}]
\label{def:weight_function_based}
$\mu_1 \dpreorder_R \mu_2$ iff there is a {\em weight} function $w :
S_1 \times S_2 \rightarrow \rat \cap [0,1]$ such that
\begin{enumerate}
  \item $\mu_1(s_1) = \sum_{s_2 \in S_2} w(s_1, s_2)$ for all $s_1 \in S_1$,
  \item $\mu_2(s_2) = \sum_{s_1 \in S_1} w(s_1, s_2)$ for all $s_2 \in S_2$,
  \item $w(s_1, s_2) > 0$ implies $s_1 R s_2$ for all $s_1 \in S_1$, $s_2
\in S_2$.
  \end{enumerate}
\end{definition}

$\mu_1 \dpreorder_R \mu_2$ can be checked by computing the maxflow in
an appropriate network and checking if it equals $1.0$ \cite{Baier_habilitation98}.
If $\mu_1 \dpreorder_R \mu_2$ holds, $w$ in the above definition is one such
maxflow function.
As explained above, $\mu_1 \dpreorder_R \mu_2$ can be understood as {\em
matching} all the probabilities (after splitting appropriately) under $\mu_1$ and $\mu_2$. Considering
$\supp{\mu_1}$ and $\supp{\mu_2}$ as two partite sets, this is the weighted analog
of saturating a partite set in bipartite matching, giving us the following
analog of the well-known Hall's Theorem for saturating
$\supp{\mu_1}$.

\begin{lemma}[\cite{Zhang_thesis08}]
\label{lem:image_based}
$\mu_1 \dpreorder_R \mu_2$ iff for every $S \subseteq \supp{\mu_1}$, $\mu_1(S)
\leq \mu_2(R(S))$.
\end{lemma}

%Intuitively, this is an analog of the {\em Hall's
%Theorem} for a maximum bipartite matching saturating a partite set in
%Graph Theory \cite{West_book} for the weighted case. 
It follows that when $\mu_1 \not\dpreorder_R \mu_2$, there exists a witness $S
\subseteq \supp{\mu_1}$ such
that $\mu_1(S) > \mu_2(R(S))$. For example, if $R(s_2) = \emptyset$ in Figure
\ref{fig:dist_rel}(a), its probability $\frac{1}{2}$ under $\mu_1$ cannot be
matched and $S = \{s_2\}$ is a witness subset.

\begin{definition}[Strong Simulation \cite{SL_nordic95}]
\label{def:strong_simulation}
$R$ is a {\em strong simulation} iff for every $s_1 R s_2$ and $s_1 \trans{a}
\mu^a_1$ there is a $\mu^a_2$ with $s_2 \trans{a} \mu^a_2$ and
$\mu^a_1 \dpreorder_R \mu^a_2$.

For $s_1 \in S_1$ and $s_2 \in S_2$, $s_2$ strongly simulates $s_1$, denoted $s_1
\preorder s_2$, iff there is a strong simulation $T$ such that $s_1 T s_2$.  
$L_2$ strongly simulates $L_1$, also denoted $L_1 \preorder L_2$, iff $s^0_1
\preorder s^0_2$. 
%For the latter, alternatively, we say that {\em simulation conformance} holds
%between $L_1$ and $L_2$.
\end{definition}

% this sentence may be removed
%As argued by Baier \cite{Baier_habilitation98}, 
%This definition is analogous to that of Milner in the
%non-probabilistic setting~\cite{Milner_tr71}.
%In this paper, we are only
%interested in checking $L_1 \preorder L_2$ when $\alpha_1 \subseteq \alpha_2$
%and if $\alpha_1 \setminus \alpha_2 \neq \emptyset$, we assume that $L_2$ is
%{\em completed} by adding Dirac self-loops on each of these extra actions from
%every state. See Figure \ref{fig:io} for an example: $P$ is 
%completed with $\{\text{{\em send}}, \text{{\em ack}}\}$ before the conformance is
%checked. 
%For example, $L \preorder P$ in Figure \ref{fig:io}.
When checking a specification $P$ of a
system $L$ with $\alpha_P \subset \alpha_L$,
we implicitly assume that $P$ is {\em completed} by adding Dirac self-loops on
each of the actions in $\alpha_L \setminus \alpha_P$ from every state before
checking $L \preorder P$. For example, $L \preorder P$ in Figure \ref{fig:io}
assuming that $P$ is completed with $\{\text{{\em send}}, \text{{\em ack}}\}$.
%Note that the algorithms described in this paper can be used, as are, even
%without this assumption. 
Checking $L_1 \preorder L_2$ is decidable in polynomial time
\cite{Baier_habilitation98,Zhang_thesis08} and can be performed with
a greatest fixed point algorithm that computes the coarsest simulation
between $L_1$ and $L_2$. The algorithm uses a relation variable $R$
initialized to $S_1 \times S_2$ and checks the condition in
Definition~\ref{def:strong_simulation} for every pair in $R$, iteratively,
removing any violating pairs from $R$. The
algorithm terminates when a fixed point is reached showing $L_1
\preorder L_2$ or when the pair of initial states is removed showing
$L_1 \not\preorder L_2$. If $n = \text{{\tt max}}(|S_1|,
|S_2|)$ and $m = \text{{\tt max}}(|\tau_1|, |\tau_2|)$,
%checking $\mu_1 \dpreorder_R \mu_2$ takes $O(n^3/\log n)$ time and
%$O(n^2)$ space and
the algorithm takes $O((mn^6 + m^2n^3)/\log
n)$ time and $O(mn + n^2)$ space \cite{Baier_habilitation98}. Several
optimizations exist \cite{Zhang_thesis08} but we do not consider them
here, for simplicity.
%We do consider a specialized fixed point algorithm for the
%case that $L_1$ is a tree, where $S_1$ is traversed bottom-up while considering
%the pairs with all the states of $S_2$, during abstraction refinement (Sections
%\ref{sec:cegar} and \ref{sec:agar}).

% this special algorithm is not a fixed-point algorithm!
We do consider a specialized algorithm for the
case that $L_1$ is a tree which we use during abstraction refinement (Sections
\ref{sec:cegar} and \ref{sec:agar}). It initializes $R$ to $S_1 \times S_2$ and is based on a bottom-up traversal of
$L_1$.  Let $s_1 \in S_1$ be a non-leaf state during such a traversal
and let $s_1 \trans{a} \mu_1$. For every $s_2 \in S_2$,
the algorithm checks if there exists $s_2 \trans{a} \mu_2$ with $\mu_1 \dpreorder_R \mu_2$
and removes $(s_1,s_2)$ from $R$, otherwise, where $R$ is the current relation.
This constitutes an iteration in the algorithm. The algorithm
terminates when $(s^0_1,s^0_2)$ is removed from $R$ or when the traversal ends.
Correctness is not hard to show and we skip the proof.
%Also, if $L_1$ is a tree, a specialized fixed point
%algorithm which traverses $S_{L_1}$ bottom-up, while considering pairs
%with all the states of $S_{L_2}$, suffices which results in a better
%algorithm.

%Below, we mention an important property of $\preorder$.

\begin{lemma}[\cite{SL_nordic95}]
\label{lem:precongruence}
$\preorder$ is a preorder $(${\em \ie}reflexive and transitive$)$ and is
compositional, {\em \ie}if $L_1 \preorder L_2$ and $\alpha_2 \subseteq
\alpha_1$, then for every LPTS $L$,
$L_1 \parallel L \preorder L_2 \parallel L$.
\end{lemma}

%Finally, the following is the assume-guarantee rule we are interested in.
%  \begin{mathpar}
%  \inferrule[ASYM]
%            {L_1 \parallel A \preorder P \\ L_2 \preorder A}
%            {L_1 \parallel L_2 \preorder P}
%  \end{mathpar}
%where $L_1$, $L_2$, $A$ and $P$ are LPTSes with $P$ standing for the {\em
%specification} which the composition $L_1 \parallel L_2$ should conform to.

Finally, we show the {\em soundness} and {\em completeness} of the rule {\sc ASym}.
The rule is {\em sound} if the conclusion holds whenever there is an $A$ satisfying
the premises. And the rule is {\em complete} if there is an $A$ satisfying the
premises whenever the conclusion holds.

\begin{theorem}
\label{thm:sound_complete}
For $\alpha_A \subseteq \alpha_2$, the rule {\sc ASym} is sound and complete.
\end{theorem}

\begin{proof}
Soundness follows from Lemma \ref{lem:precongruence}. Completeness
follows trivially by replacing $A$ with $L_2$.  \qed
\end{proof}

%{\bf remind the reader why we need completeness}

%(Un)Fortunately, counterexamples to the conformance have not been
%previously characterized. And it is a natural question to ask as
%counterexamples are crucial to check the conformance for huge systems by means
%of abstraction refinement, or otherwise, as has also been suggested as a future work in a
%recent PhD thesis \cite{Zhang_thesis08}. Even at a more fundamental level, one
%would expect a counterexample when the conformance fails, to know a reason behind
%the failure. This is the main content of Section \ref{sec:cex}.

\section{Counterexamples to Strong Simulation}
\label{sec:cex}
\vspace{-0.1in}
Let $L_1$ and $L_2$ be two LPTSes. We characterize a counterexample to
$L_1 \preorder L_2$ as a tree and show
that any simpler structure is not sufficient in general. We first
describe counterexamples via a simple language-theoretic
characterization.

\begin{definition}[Language of an LPTS]
Given an LPTS $L$, we define its language, denoted $\L(L)$, as the set
$\{L' | L' ~\text{is an LPTS and } L' \preorder L\}$.
\end{definition}

\begin{lemma}
$L_1 \preorder L_2$ iff $\L(L_1) \subseteq \L(L_2)$.
\end{lemma}

\begin{proof}
Necessity follows trivially from the transitivity of $\preorder$ and sufficiency
follows from the reflexivity of $\preorder$ which implies $L_1 \in \L(L_1)$.
\qed
\end{proof}

Thus, a counterexample $C$ can be defined as follows.

\begin{definition}[Counterexample]
\label{def:cex}
A counterexample to $L_1 \preorder L_2$ is an LPTS $C$ such that $C \in \L(L_1)
\setminus \L(L_2)$, i.e. $C \preorder L_1$ but $C \not\preorder L_2$.
\end{definition}

Now, $L_1$ itself is a trivial choice for $C$ but it does not give any more useful
information than what we had before checking the simulation.
Moreover, it is preferable to have $C$
with a special and simpler structure rather than a general LPTS as it helps in a more
efficient counterexample analysis, wherever it is
used. When the LPTSes are restricted to LTSes, a {\em tree-shaped} LTS is known
to be sufficient as a counterexample \cite{Chaki_thesis05}. Based on a similar
intuition, we show that a stochastic tree is sufficient as a counterexample in the probabilistic case.

\begin{theorem}
\label{thm:tree_cex}
If $L_1 \not\preorder L_2$, there is a tree which serves as a counterexample.
\end{theorem}

\begin{proof}
We only give a brief sketch of a constructive proof here. See Appendix for a detailed proof.
Counterexample generation is based on the
coarsest strong simulation computation from Section
\ref{sec:prelims}. By induction on the number of pairs not in the current relation
$R$, we show that there is a tree counterexample to $s_1 \preorder
s_2$ whenever $(s_1, s_2)$ is removed from $R$. We only consider the
inductive case here. The pair is removed because there is a transition
$s_1 \trans{a} \mu_1$ but for every $s_2 \trans{a} \mu$,
$\mu_1 \not\dpreorder_R \mu$ \ie there exists $S^\mu_1
\subseteq \supp{\mu_1}$ such that $\mu_1(S^\mu_1) >
\mu(R(S^\mu_1))$. Such an $S^\mu_1$ can be found using
Algorithm \ref{algo:dist_cex}. Now, no pair in $S^\mu_1 \times
(\supp{\mu} \setminus R(S^\mu_1))$ is in $R$. By induction
hypothesis, a counterexample tree exists for each such pair. A
counterexample to $s_1 \preorder s_2$ is built using $\mu_1$ and
all these other trees.  \qed
\end{proof}

\vspace{-0.2in}
\begin{algorithm}
\footnotesize
\caption{Finding $T \subseteq S_1$ such that $\mu_1(T) > \mu(R(T))$.}
\label{algo:dist_cex}
$\text{Given $\mu_1 \in \dist{S_1}$, $\mu \in \dist{S_2}$, $R
\subseteq S_1 \times S_2$ with $\mu_1 \not\dpreorder_R \mu$.}$
%\vspace{-0.15in}
\begin{algorithmic}[1]
%\STATE build the flow network and compute the maxflow $f$
%\cite{Baier_habilitation98}
\STATE let $f$ be a maxflow function for the flow network corresponding to
$\mu_1$ and $\mu$
\STATE find $s_1 \in S_1$ with $\mu_1(s_1) > \sum_{s_2 \in S_2} f(s_1,s_2)$ and
let $T = \{s_1\}$
%  \COMMENT{{\em an unsaturated vertex}}
\WHILE {$\mu_1(T) \le \mu(R(T))$}
  \STATE $T \leftarrow \{s_1 \in S_1 | \exists s_2 \in R(T) : f(s_1,s_2) > 0\}$
\ENDWHILE
\RETURN $T$
\end{algorithmic}
\end{algorithm}
\vspace{-0.2in}

For an illustration, see Figure \ref{fig:cex_eg} where $C$ is a counterexample
to $L_1 \preorder L_2$. Algorithm \ref{algo:dist_cex} is also analogous to the one
used to find a subset failing Hall's condition in Graph Theory and can easily be
proved correct. We obtain the following complexity bounds whose proof can be found in Appendix.

%\begin{figure}[t]
%\hspace{0.1in}
%\begin{minipage}[b]{0.3\linewidth}
%%\begin{figure}
%\centering
%\includegraphics[scale=1.5]{figures/tree_eg.eps}
%\caption{An example tree.}
%\label{fig:tree_eg}
%%\end{figure}
%\end{minipage}
%\hspace{0.1in}
%\begin{minipage}[b]{0.5\linewidth}
%%\begin{figure}
%\centering
%\includegraphics[scale=1.4]{figures/cex_eg.eps}
%\caption{$C$ is a counterexample to $L_1 \preorder L_2$.}
%\label{fig:cex_eg}
%%\end{figure}
%\end{minipage}
%\vspace{-0.3in}
%\end{figure}

\begin{figure}[t]
\centering
\includegraphics[scale=1.4]{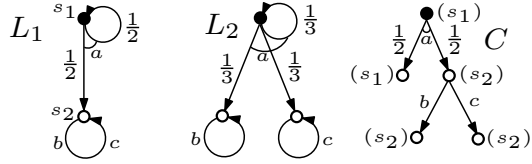}
\caption{$C$ is a counterexample to $L_1 \preorder L_2$.}
\label{fig:cex_eg}
\vspace{-0.2in}
\end{figure}

\begin{theorem}
\label{thm:cex_complexity}
Deciding $L_1 \preorder L_2$ and obtaining a tree counterexample
takes $O(mn^6 + m^2n^3)$ time and $O(mn + n^2)$ space where $n = \text{{\tt
max}}(|S_{L_1}|, |S_{L_2}|)$ and $m = \text{{\tt max}}(|\tau_1|, |\tau_2|)$.
\end{theorem}

Note that the obtained counterexample is essentially a finite {\em tree execution} of
$L_1$. That is, there is a total mapping $M : S_C \to S_1$ such that for every
transition $c \trans{a} \mu_c$ of $C$, there exists $M(c) \trans{a} \mu_1$ such that $M$ restricted
to $\supp{\mu_c}$ is an injection and for every $c' \in \supp{\mu_c}$,
$\mu_c(c') = \mu_1(M(c'))$. $M$ is also a strong simulation. We call
such a mapping an {\em execution mapping from $C$ to $L_1$}. Figure
\ref{fig:cex_eg} shows an execution mapping in brackets beside the states of $C$.
We therefore have the following corollary.

\begin{corollary}
\label{cor:reactive_tree_cex}
If $L_1$ is reactive and $L_1 \not\preorder L_2$, there is a reactive
tree which serves as a counterexample.
\end{corollary}

The following two lemmas show that (reactive) trees are the simplest
structured counterexamples (proofs in Appendix).
%The proofs, given in Appendix, use
%the LPTSes in Figures \ref{fig:cex_fully_prob_insuff} and
%\ref{fig:cex_reactive_insuff}, respectively.

\begin{lemma}
\label{lem:cex_fully_prob_insuff}
%There exist reactive LPTSes $R_1$ and $R_2$ such that no fully-probabilistic
%LPTS serves as a counterexample to $R_1 \preorder R_2$.
There exist reactive LPTSes $R_1$ and $R_2$ such that $R_1 \not\preorder R_2$
and no counterexample is fully-probabilistic.
\end{lemma}

%\begin{figure}
%\centering
%\includegraphics[scale=1.2]{figures/cex_fully_prob_insuff.eps}
%\caption{An example where there is no fully-probabilistic counterexample.}
%\label{fig:cex_fully_prob_insuff}
%\end{figure}

%For example, both the transitions on actions $y$ and $z$ from state $r_{11}$ of
%$R_1$ are necessary to show that $R_1 \not\preorder R_2$ in Figure
%\ref{fig:cex_fully_prob_insuff}. The proof in Appendix has more details.
Thus, if $L_1$ is reactive, a reactive tree %(Corollary \ref{cor:reactive_tree_cex}) is %possibly 
is the simplest structure for a
counterexample to $L_1 \preorder L_2$.
%rewrite this as "unlike det. LTSes, probabilities also result in branching in
%the cex" ?
This is surprising, since the non-probabilistic counterpart of a
fully-probabilistic LPTS is a trace of actions and it is known that trace inclusion coincides with
simulation conformance between reactive (\ie deterministic) LTSes.
%citation?
If there is no such restriction on $L_1$, one may ask if a reactive LPTS suffices
as a counterexample to $L_1 \preorder L_2$. That is not the case either,
as the following lemma shows.

\begin{lemma}
\label{lem:cex_reactive_insuff}
%There exist an LPTS $L$ and a reactive LPTS $R$ such that no reactive LPTS serves as a
%counterexample to $L \preorder R$.
There exist an LPTS $L$ and a reactive LPTS $R$ such that $L \not\preorder R$
and no counterexample is reactive.
\end{lemma}

%\begin{figure}
%\begin{minipage}[b]{0.5\linewidth}
%%\begin{figure}
%\centering
%\includegraphics[scale=0.9]{figures/cex_reactive_insuff.eps}
%\caption{There is no reactive counterexample to $L \preorder R$.}
%\label{fig:cex_reactive_insuff}
%%\end{figure}
%\end{minipage}
%\hspace{0.1cm}
%\begin{minipage}[b]{0.5\linewidth}
%%\begin{figure}
%\centering
%\includegraphics[scale=1.3]{figures/quotient}
%%\caption{An LPTS $L$, partition $\Pi = \{c_1,c_2\}$ and the quotient $L/\Pi$.}
%\caption{$L/\Pi$ for a partition $\Pi = \{c_1,c_2\}$.}
%\label{fig:quotient}
%%\end{figure}
%\end{minipage}
%\end{figure}

%For example, both the transitions on action $y$ from the state $l_{11}$ of
%$L$ are necessary to show that $L \not\preorder R$ in Figure
%\ref{fig:cex_reactive_insuff}. The proof in Appendix has more details.

%Thus, if $L_1$ is a general LPTS, a tree %(Theorem \ref{thm:tree_cex}) is possibly 
%is the simplest structured countereaxmple to $L_1 \preorder L_2$.
%Thus, for (reactive) LPTSes, a (reactive) tree is probably the simplest structured
%counterexample to strong simulation conformance. To summarize, for any reactive
%LPTSes $R_1$ and $R_2$ and LPTSes $L_1$ and $L_2$, $R_1 \not\preorder R_2$ or
%$R_1 \not\preorder L_1$ can always be witnessed by a {\em reactive} tree but need not
%be by a fully probabilistic LPTS and $L_1 \not\preorder R_1$ or $L_1 \not\preorder L_2$ can
%always be witnessed by a tree but need not be by a {\em reactive} tree.

\section{CEGAR for Checking Strong Simulation}
\label{sec:cegar}
\vspace{-0.1in} Now that the notion of a counterexample has been
formalized, we describe a CounterExample Guided Abstraction Refinement
(CEGAR) approach \cite{CGJ+_cav00} to check $L \preorder P$ where $L$
and $P$ are LPTSes and $P$ stands for a {\em specification} of $L$. We
will use this approach to describe AGAR in the next section.

Abstractions for $L$ are obtained using a quotient construction from a
{\em partition} $\Pi$ of $S_L$.  We let $\Pi$ also denote the
corresponding set of equivalence classes and given an arbitrary
$s \in S$, let $[s]_\Pi$ denote the equivalence class containing
$s$. The quotient is an adaptation of the usual construction in the
non-probabilistic case.

%\begin{definition}[Partition]
%\label{def:partition}
%A partition $\Pi$ of a set $S$ is a pair $(E,[\cdot])$ where $E$ is a
%set of disjoint subsets of $S$ with $\bigcup E = S$ and $[\cdot] : S
%\to E$ is a surjective total function such that for every $s \in S$,
%$s \in [\cdot](s)$. For convenience, we refer to a member of $E$ as an
%    {\em equivalence class} and we abuse the notation to use $\Pi$ to
%    denote $E$ and $[s]_\Pi$ to denote $[\cdot](s)$ or simply $[s]$
%    when $\Pi$ is clear.
%\end{definition}

\begin{definition}[Quotient LPTS]
\label{def:quotient_lpts}
Given a partition $\Pi$ of $S_L$, define the {\em quotient LPTS}, denoted
$L/\Pi$, as the LPTS $\langle \Pi,[s^0_L]_\Pi,\alpha_L,\tau \rangle$ where
$(c,a,\mu_l) \in \tau$ iff $(s,a,\mu) \in \tau_L$ for some $s \in S_L$
with $s \in c$ and $\mu_l(c') = \sum_{t \in c'} \mu(t)$ for all $c' \in \Pi$.
\end{definition}

As the abstractions are built from an explicit representation of $L$,
this is not immediately useful. But, as we will see in Sections
\ref{sec:agar} and \ref{sec:results}, this becomes very useful when adapted
to the assume-guarantee setting.

\begin{figure}[t]
\begin{minipage}[c]{0.5\linewidth}
\centering
\includegraphics[scale=1.5]{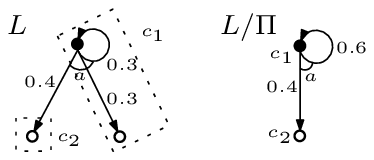}
\caption{An LPTS $L$, partition $\Pi = \{c_1,c_2\}$ and the quotient $L/\Pi$.}
\label{fig:quotient}
\end{minipage}
\hspace{0.1in}
\begin{minipage}[c]{0.5\linewidth}
\centering
\includegraphics[scale=1.8]{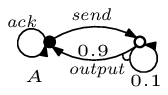}
\caption{An assumption for $L_1$, $L_2$ and $P$\newline  in Figure \ref{fig:io}.}
\label{fig:assumption}
\end{minipage}
\vspace{-0.2in}
\end{figure}

Figure \ref{fig:quotient} shows an example quotient.
Note that $L \preorder L/\Pi$ for any partition $\Pi$ of $S_L$ (proof in
Appendix), with the relation $R = \{(s,c) | s \in c, c \in \Pi\}$ as a strong simulation.
%Thus, we are only interested in abstractions which are {\em quotient LPTS}es
%w.r.t. partitions of $S_L$. 

\begin{algorithm}
\footnotesize
\caption{CEGAR for LPTSes: checks $L\preorder P$}
\label{algo:cegar}
\begin{algorithmic}[1]
\STATE $A \leftarrow L/\Pi$, where $\Pi$ is the coarsest partition of $S_L$
\WHILE {$A \not\preorder P$}
  \STATE obtain a counterexample $C$
  \STATE $(spurious,A')\leftarrow \text{{\em analyzeAndRefine}}(C,A,L)$
\COMMENT{{\em see text}}
  \IF {$spurious$}
    \STATE $A \leftarrow A'$
  \ELSE
    \RETURN {counterexample $C$}
  \ENDIF
\ENDWHILE
\RETURN $L\preorder P$ holds
\end{algorithmic}
\end{algorithm}

CEGAR for LPTSes is sketched in
Algorithm~\ref{algo:cegar}.
%It starts with the smallest abstraction $A$, which is
It maintains an abstraction $A$ of $L$, initialized to
the quotient for the coarsest partition, and iteratively
refines %the current abstraction
$A$ based on the counterexamples obtained
from the simulation check against $P$ until a partition whose
corresponding quotient conforms to $P$ w.r.t. $\preorder$ is obtained, or a real
counterexample is found. In the following, we describe how to analyze
if a counterexample is {\em spurious}, due to abstraction, 
and how to refine the abstraction in case it is (lines
$4-6$).  Our analysis is an adaptation of an existing one for
counterexamples which are arbitrary {\em sub-structures}
%(\ie
%$I$ maps at most one state in $S_C$ to any state of $S_{L/\Pi}$) 
of $A$ \cite{CV_tocl10}; while our tree counterexamples have an execution
mapping to $A$, they are not necessarily sub-structures of $A$.

\vspace{0.1in}
\noindent{\bf Analysis and Refinement ({\em analyzeAndRefine}$(\cdot)$).}
Assume that $\Pi$ is a partition of $S_L$ such that $A=L/\Pi$ and
$A\not\preorder P$. Let $C$ be a tree counterexample obtained by the
algorithm described in Section \ref{sec:cex}, \ie $C \preorder A$ but
$C \not\preorder P$. As described in Section \ref{sec:cex}, there is an {\em
execution mapping} $M : S_C \to S_A$ which is also
a strong simulation. Let $R_M \subseteq S_C \times S_L$ be $\{(s_1,s_2) | s_1 M [s_2]_\Pi\}$.
%$s_1 R_I s_2$ iff $s_1 I [s_2]_\Pi$. %When checking $C\preorder L$, 
Our refinement strategy tries to
obtain the coarsest strong simulation between $C$ and $L$ contained in
$R_M$, using the specialized algorithm for trees described
in Section \ref{sec:prelims} with $R_M$ as the initial candidate.
%As $C \not\preorder L$ we know that no strong simulation contains
%$(s^0_C, s^0_L)$.
Let $R$ and $R_\old$ be the candidate relations at the end of the
current and the previous iterations, respectively, and let
$s_1 \trans{a} \mu_1$ be the transition in $C$ considered by the
algorithm in the current iteration. ($R_\old$ is undefined initially.)  The strategy
refines a state when one of the following two cases happens before termination
and otherwise, returns $C$ as a {\em real} counterexample. 

\begin{enumerate}

\item $R(s_1) = \emptyset$. There are two possible reasons for this case. One is
that the states in $\supp{\mu_1}$ are not related, by $R$, to enough number
of states in $S_L$ (\ie $\mu_1$ is {\em spurious}) and (the images under $M$
of) all the states in $\supp{\mu_1}$ are candidates for refinement. The other
possible reason is the branching (more than one transition) from $s_1$ where no
state in $R_M(s_1)$ can {\em simulate all} the transitions of $s_1$ and $M(s_1)$
is a candidate for refinement.
%This case can arise because the states in
%  $\supp{\mu_1}$ are not related (by $R$) to enough number of
%  states in $S_L$, and therefore $\mu_1$ is "spurious".
%% and (the images under $I$ of) $\supp{\mu_1}$
%%  correspond to the {\em dead-end} states in classical CEGAR for
%%  traces {\bf who knows about these dead end states???}. 
%  $\supp{\mu_1}$ can have more than one state, and all of them are candidates for refinement. 
%  This case can also be solely due to branching (more than one
%  transition) in $s_1$ where no state in $R_I(s_1)$ can {\em simulate
%    all} the transitions of $s_1$.  {\bf what are the candidates for refinement in this case?}

%This wasn't possible in the
%  classical trace setting as there was no branching.

%The strategy splits every equivalence class $c$ in $\supp{\mu'_1}$
%into the two groups of states $R(I^{-1}(c))$ and the rest. This is analogous to
%splitting a dead-end abstract state in classical CEGAR for traces.

%\item $s^0_L \not\in R(s^0_C)$. At this point, we also obtain the transition
%under consideration and split the states in the support of the corresponding
%sub-distribution in $L$ along with splitting $I(s^0_C)$ in the same fashion as
%above. This is analogous to splitting the initial state in classical CEGAR.

\item $M(s_1) = [s^0_L]_\Pi$, $s^0_L \in R_\old(s_1) \setminus R(s_1)$ and
$R(s_1) \neq \emptyset$, \ie $M(s_1)$ is the initial state of $A$ but $s_1$ is no
longer related to $s^0_L$ by $R$. Here, $M(s_1)$ is a candidate for refinement.

\end{enumerate}

%Similar to classical CEGAR
In case $1$, our refinement strategy first tries to split the equivalence class $M(s_1)$ into
$R_\old(s_1)$ and the rest and then, for every state $s \in \supp{\mu_1}$, tries
to split the equivalence class $M(s)$ into $R_\old(s)$ and the rest, unless
$M(s) = M(s_1)$ and $M(s_1)$ has already been split.
And in case $2$, the strategy splits the equivalence class $M(s_1)$ into
$R_\old(s_1) \setminus R(s_1)$ and the rest. It follows from the two cases that
if $C$ is declared real, then $C \preorder L$ with the final $R$ as a strong simulation between $C$ and
$L$ and hence, $C$ is a counterexample to $L \preorder P$.
The following lemma (proof in Appendix) shows that the refinement strategy always leads to progress.
%represented by  $I(s_1)$ into 
%$R_\old(s_1) \setminus R(s_1)$ and the rest. It may be the case that  $R_\old(s_1) =
%R_I(s_1)$ and $R(s_1) = \emptyset$, in which case, there is no splitting.
%Furthermore, the strategy splits
%each subset represented by $I(s)$ for $s \in \supp{\mu}$ into $R_\old(s)$ and the
%rest; it may be the case that  $I(s) = I(s_1)$, in which case $I(s_1)$ has already been
%split.  The obtained state partition is used to build a new abstraction for $L$.
%The following lemma, whose proof can be found in Appendix,
%shows that the above strategy always leads to a refinement.

\begin{lemma}
\label{lem:ref_progress}
The above refinement strategy always results in a strictly finer partition $\Pi'
< \Pi$.
%split for at least one state of $L/\Pi$ in both the above cases. {\bf what does it mean "non-trivial" split?}
\end{lemma}

\section{Assume-Guarantee Abstraction Refinement}
\label{sec:agar}
\vspace{-0.1in} 
We now describe our approach to Assume-Guarantee
Abstraction Refinement (AGAR) for LPTSes.  The approach is similar to CEGAR
from the previous section with the notable
exception that counterexample analysis is performed in an assume
guarantee style: a counterexample obtained from checking one component
is used to refine the abstraction of a different component.

Given LPTSes $L_1$, $L_2$
and $P$, the goal is to check $L_1 \parallel L_2 \preorder P$ in an
assume-guarantee style, using rule {\sc ASym}.
%We would like to
%use the tree counterexamples to show conformance in the compositional setting of
%the rule ASYM mentioned in Section \ref{sec:prelims}.
%This approach is based on CEGAR adapted to the compositional setting,
%similar to the one in a classical trace inclusion
%setting~\cite{BPG_cav08}.
The basic idea is to maintain $A$ in the rule as an abstraction of
$L_2$, \ie the second premise holds for free throughout, and to
%use $A$ to iteratively check the first premise.
check only the first premise for every $A$ generated by the algorithm.
As in CEGAR, we restrict $A$ to the quotient for a partition of $S_2$.
If the first premise holds for an $A$, then $L_1 \parallel L_2\preorder P$ also
holds, by the soundness of
the rule. Otherwise, the obtained counterexample $C$ is analyzed to
see whether it indicates a real error or it is spurious, in which case $A$
is refined (as described in detail below). Algorithm \ref{algo:agar}
sketches the AGAR loop. 

For an example, $A$ in Figure~\ref{fig:assumption} shows the
final assumption generated by AGAR for the LPTSes in Figure
\ref{fig:io} (after one refinement). 

%\begin{algorithm}
%\footnotesize
%\caption{AGAR for LPTSes: checks $L_1 \parallel L_2\preorder P$}
%\label{algo:agar}
%\begin{algorithmic}[1]
%\STATE $A \leftarrow$ coarsest abstraction of $L_2$
%\WHILE {$L_1 \parallel A \not\preorder P$}
%  \STATE obtain a counterexample $C$
%  \STATE obtain projections $C \project{L_1}$ and $C \project{A}$
%  \IF[\textit{this is only used with {\sc ASym-N}}] {$(C \project{L_1})^{\alpha_1} \parallel L_2 \not\preorder P$}
%      
%    \RETURN {a counterexample $C'$}
%  \ELSE
%    \STATE $(\_,A)\leftarrow  \text{{\em analyzeAndRefine}}(C \project{A},A,L_2)$
%  \ENDIF
%\ENDWHILE
%\RETURN $L_1 \parallel L_2\preorder P$ holds
%\end{algorithmic}
%\end{algorithm}

\vspace{-0.2in}
\begin{algorithm}
\footnotesize
\caption{AGAR for LPTSes: checks $L_1 \parallel L_2\preorder P$}
\label{algo:agar}
\begin{algorithmic}[1]
\STATE $A \leftarrow$ coarsest abstraction of $L_2$
\WHILE {$L_1 \parallel A \not\preorder P$}
  \STATE obtain a counterexample $C$
  \STATE obtain projections $C \project{L_1}$ and $C \project{A}$
  \STATE $(\text{\em spurious}, A') \leftarrow \text{{\em analyzeAndRefine}}(C \project{A},A,L_2)$
  \IF {{\em spurious}}
    \STATE $A \leftarrow A'$
  \ELSE
    \RETURN {counterexample $C$}
  \ENDIF
\ENDWHILE
\RETURN $L_1 \parallel L_2\preorder P$ holds
\end{algorithmic}
\end{algorithm}
\vspace{-0.2in}

%The abstraction refinement algorithm starts with $A$ as the
%quotient for the coarsest partition of $L_2$.  To verify if $A$ is
%sufficient to show the conformance in the conclusion of the rule, it
%suffices to check the first premise $L_1 \parallel A \preorder P$. If
%it holds, the conclusion also holds by the soundness of rule {\sc
%  ASYM}. Otherwise, we obtain a counterexample $C$ to the first
%premise. Below, we describe how to check if $C$ is spurious and a
%refinement heuristic in case it is.

\noindent{\bf Analysis and Refinement.}
%A na\"{i}ve approach is to obtain a corresponding partition $\Pi'$ of
%$L_1 \parallel L_2$ from $\Pi$ and do the analysis described in
%Section \ref{sec:cegar} between $C$ and $(L_1 \parallel
%L_2)/\Pi'$. But that amounts to considering the state space of $L_1
%\parallel L_2$ which is exactly what we want to avoid by means of
%assume-guarantee verification.
%Ideally, to check if $C$ is spurious, we will check $C \preorder L_1
%\parallel L_2$. But that amounts to constructing the composition of $L_1$ and
%$L_2$ which is exactly what we want to avoid.
The counterexample analysis is performed compositionally, using the {\em
  projections} of $C$ onto $L_1$ and $A$. 
%As there is a one-to-one correspondence between the states and sub-distributions of $C$ and
As there is an {\em execution mapping} from $C$ to
$L_1 \parallel A$, these projections are the {\em
  contributions} of $L_1$ and $A$ towards $C$ in the composition.  We
denote these projections by $C \project{L_1}$ and $C \project{A}$,
respectively.
In the non-probabilistic case, these are obtained by simply projecting $C$ onto
the respective alphabets.
%We note that these {\em contributions} are not obtained by simply projecting $C$ onto the
%respective alphabets, as it is typical for classical,
%non-probabilistic LTSes.
In the probabilistic scenario, however, composition
changes the probabilities in the distributions (Definition
\ref{def:product}) and alphabet projection is insufficient.
For this reason, we additionally record the
%We need to maintain additional information recording the
individual distributions of the LPTSes responsible for a product
distribution while performing the composition. Thus, projections
$C \project{L_1}$ and $C \project{A}$ can be obtained using this
auxiliary information.
%Another crucial difference is that in
%our case, we can only conclude that $C \preorder C \project{L_1} \parallel C
%\project{A}$ and the other direction need not always hold. The former can easily
%be proved by showing that the trivial injection is a strong simulation. For the
%latter, one easy counterexample is when $L_1$ and $A$ only have distributions while $C$ has
%sub-distributions.
%Even if projections are allowed to have sub-distributions,
%there can be cases where the other direction fails.
%Note that $C \preorder (C \project{L_1})^{\alpha_1} \parallel (C
%\project{A})^{\alpha_2}$. Also,
%$C \project{L_1} \preorder L_1$ and $C \project{A} \preorder A$ and hence,
%$(C \project{L_1})^{\alpha_1} \preorder L_1$ and $(C \project{A})^{\alpha_2} \preorder A$. Therefore,
%Lemma \ref{lem:precongruence} implies that $(C \project{L_1})^{\alpha_1}
%\parallel (C \project{A})^{\alpha_2}$ is another counterexample to the first premise.
Note that there is a natural {\em execution mapping} from $C \project{A}$ to
$A$ and from $C \project{L_1}$ to $L_1$. We can then employ the analysis
described in Section~\ref{sec:cegar} between $C \project{A}$ and $A$, \ie
invoke $\text{{\em analyzeAndRefine}}(C \project{A},A,L_2)$ to determine if $C
\project{A}$ (and hence, $C$) is spurious and to refine $A$ in case it is.
Otherwise, $C \project{A} \preorder L_2$ and hence, $(C \project{A})^{\alpha_2}
\preorder L_2$. Together with $(C \project{L_1})^{\alpha_1} \preorder L_1$ this
implies $(C \project{L_1})^{\alpha_1} \parallel (C \project{A})^{\alpha_2}
\preorder L_1 \parallel L_2$ (Lemma \ref{lem:precongruence}). As $C \preorder
(C \project{L_1})^{\alpha_1} \parallel (C \project{A})^{\alpha_2}$, $C$ is then
a {\em real} counterexample. Thus, we have the following result.

\begin{theorem}[Correctness and Termination]
Algorithm AGAR always terminates with at most
$|S_2| - 1$ refinements and $L_1 \parallel L_2
\not\preorder P$ if and only if the algorithm returns a real counterexample.
\end{theorem}
\begin{proof}
{\em Correctness}: AGAR terminates when either Premise 1 is satisfied by the
current assumption (line $12$) or when a counterexample is returned (line $9$).
In the first case, we know that Premise 2 holds by construction and since {\sc
ASym} is sound (Theorem \ref{thm:sound_complete}) it follows that indeed $L_1
\|L_2 \preorder P$. In the second case, the counterexample returned by AGAR is
real (see above) showing that $L_1 \parallel L_2 \not\preorder P$.
%AGAR returns a counterexample $C$ such that $C\not\preorder P$ and $C \preorder
%(C \project{L_1})^{\alpha_1} \parallel (C \project{A})^{\alpha_2}$. Since  $(C \project{L_1})^{\alpha_1} \parallel (C \project{A})^{\alpha_2}
%\preorder L_1 \parallel L_2$ (see above) it follows that indeed $C$ is a real counterexample. 

{\em Termination}: AGAR iteratively refines the abstraction until the property
holds or a real counterexample is reported. Abstraction
refinement results in a finer partition (Lemma~\ref{lem:ref_progress}) and thus
it is guaranteed to terminate since in the worst case $A$ converges to $L_2$
which is finite state. Since rule {\sc ASym} is trivially complete for $L_2$
(proof of Theorem \ref{thm:sound_complete}) it follows that AGAR will also
terminate, and the number of refinements is bounded by $|S_2|-1$.
\qed
\end{proof}

In practice, we expect AGAR to terminate earlier than in $|S_2|-1$ steps, with an assumption smaller than $L_2$. AGAR will terminate as soon as it finds an assumption that satisfies the premises or that helps exhibit a real counterexample.
Note also that, although AGAR uses an explicit representation for the individual
components, it never builds $L_1 \parallel L_2$ directly (except in the
worst-case) keeping the cost of verification low.

\vspace{0.1in}
\noindent{\bf Reasoning with $n \ge 2$ Components.}
So far, we have discussed compositional verification in the context of
two components $L_1$ and $L_2$. This  reasoning
can be generalized to $n \ge 2$ components using the
following (sound and complete) rule.
\vspace{-0.01in}
 \begin{mathpar}
  \inferrule*[right=$(${\sc ASym-N}$)$]
            {1: L_1 \parallel A_1 \preorder P \\ 2: L_2 \parallel A_2 \preorder A_1 \\ ... \\ n: L_n \preorder A_{n-1}}
            %{L_1 \parallel L_2 \parallel ... L_n \preorder P}
            {\parallel_{i=1}^n L_i \preorder P}
  \end{mathpar}
\vspace{-0.2in}

%Assumption $A_1$ can be computed using the heuristic above; $A_2$ can
%be again computed using the heuristic, where $A_1$ is used as the
%specification, etc. The last assumption $A_{n-1}$ computed for
%$M_{n-1}$ with specification $A_{n-2}$ is simply discharged on $M_n$.

The rule enables us to overcome the {\em intermediate state explosion} that
may be associated with two-way decompositions (when the subsystems are
larger than the entire system). The AGAR algorithm for this rule
involves the creation of $n-1$ nested instances of AGAR for two components,
%to compute $A_i$ as an abstraction of
%$L_{i+1} \parallel A_{i+1}$, for $i < n-1$, and $A_{n-1}$ as an abstraction
%of $L_n$.
with the $i$th instance computing the assumption $A_i$ for
$(L_1 \parallel \dots \parallel L_i) \parallel (L_{i+1} \parallel A_{i+1})
\preorder P$.
When the AGAR instance for $A_{i-1}$
returns a counterexample $C$, for $1 < i \le n-1$, we need to analyze $C$
for spuriousness and refine $A_i$ in case it is.
From Algorithm \ref{algo:agar}, $C$ is returned only if {\em
analyzeAndRefine}$(C \project{A_{i-1}},A_{i-1},L_i \parallel A_i)$
concludes that $C \project{A_{i-1}}$ is real (note that $A_{i-1}$ is an
abstraction of $L_i \parallel A_i$). From {\em
analyzeAndRefine} in Section \ref{sec:cegar}, this implies that the final
relation $R$ computed between the states of $C
\project{A_{i-1}}$ and $L_i \parallel A_i$ is a
strong simulation between them. It follows that, although $C \project{A_{i-1}}$ does not
have an {\em execution mapping} to $L_i \parallel A_i$, we can naturally obtain a tree $T$
using $C \project{A_{i-1}}$, via $R$, with such a mapping. Thus, we modify the
algorithm to return $T \project{A_i}$ at line $9$, instead of $C$, which can then be used to check for spuriousness and
refine $A_i$. Note that when $A_i$ is refined, all the $A_j$'s for $j<i$ need to
be recomputed.% in case it is.

%Note that the procedure {\em
%analyzeAndRefine} (Section \ref{sec:cegar}) needs a counterexample with an {\em execution mapping}
%to $A_i$.
%Now, an {\em execution mapping} from $C \project{A_{i-1}}$ to $L_i
%\parallel A_i$ would have sufficed as we would have projected that onto $A_i$
%for further analysis. Unfortunately, such a mapping does not exist when
%Algorithm \ref{algo:agar} is used. Therefore, we modify the algorithm as
%described in the above paragraph which generates a counterexample with such a
%mapping to $A_i$.

%This can be achieved if $C$ has a similar correspondence with $L_{i} \parallel A_{i}$ which can then be projected
%onto $A_i$. This is the reason we need the additional check during
%counterexample analysis, mentioned previously.
%Counterexamples from checking premise 1 are used to refine the
%intermediate abstractions $A_1, ... A_{n-1}$ {\bf explain projections;
%  the need for the one-to-one correspondance etc}. When $A_i$ is
%refined, them $A_1, ... A_{i-1}$ are refined as well to eliminate a
%spurios counterexample. {\bf to be updated}.

\vspace{0.1in}
\noindent{\bf Compositional Verification of Logical Properties.}
AGAR can be further applied to automate assume-guarantee checking of properties 
$\phi$ written as formulae in a logic that is
preserved by strong simulation
such as the
{\em weak-safety} fragment of probabilistic CTL (pCTL)
\cite{CV_tocl10} %and {\em probabilistic safety} properties
%\cite{FKP_fase11,FHK+_atva11}, 
which also yield trees as counterexamples.
%\footnote{These
%counterexamples are originally characterized as finite sets of finite traces
%\cite{HKB_tse09} which can also be seen as trees.}.
%The modified rule is both sound and complete for this logic.
The rule {\sc ASym} is both sound and complete for this logic ($\models$ denotes
property satisfaction) for $\alpha_A \subseteq \alpha_2$ with a proof similar to
that of Theorem \ref{thm:sound_complete}.

\vspace{-0.15in}

  \begin{mathpar}
  \inferrule
            {1: L_1 \parallel A \models \phi \\ 2: L_2 \preorder A}
            {L_1 \parallel L_2 \models \phi}
  \end{mathpar}
\vspace{-0.1in}

$A$ can be computed as a conservative abstraction of $L_2$ and
iteratively refined based on the tree counterexamples to premise $1$, using
the same procedures as before. 
%Thus, such a compositional approach is guaranteed
%to terminate unlike existing methods \cite{FHK+_atva11}.
The rule can be generalized to reasoning about $n \ge 2$ components as
described above and also to richer logics with more general counterexamples adapting
existing CEGAR approaches \cite{CV_tocl10} to AGAR.
We plan to further investigate this direction in the future.
%
%{\bf perhaps cut this para?}
%Say it can be used to ``query'' the system (obtain min max bounds on
%probabilities, similar to Oxfords' group work).  However, the work
%here is more general, since the properties are more general.

%{\bf add sub-section that discusses the explicit nature of our approach; how we
%still perform abstraction for one component at a time and we never analyze the
%explict state of the whole composition. Also, emphasize that we are not doing
%AGAR because our CEGAR is explicit-state; but because both are independent
%approaches - and we can make both approaches symbolic, in principle!}

\section{Implementation and Results}
\label{sec:results}
%\vspace{-0.1in}
\noindent{\bf Implementation.} 
We implemented the algorithms for checking simulation (Section
\ref{sec:prelims}), for generating counterexamples (as in the proof of Lemma
\ref{thm:tree_cex}) and for AGAR (Algorithm \ref{algo:agar}) with {\sc ASym} and
{\sc ASym-N} in \java.  We used the front-end of PRISM's
\cite{KNP_cav11} explicit-state engine to parse the models of the components described in
PRISM's input language and construct LPTSes which were then handled by
our implementation.

While the \java implementation for checking simulation uses the
greatest fixed point computation to obtain the coarsest strong
simulation, we noticed that the problem of checking the existence of a
strong simulation is essentially a constraint satisfaction problem. To
leverage the efficient constraint solvers that exist today, we reduced
the problem of checking simulation to an SMT problem with rational
linear arithmetic as follows. For every pair of states, the constraint
that the pair is in some strong simulation is simply the encoding of
the condition in Definition \ref{def:strong_simulation}. For a
relevant pair of distributions $\mu_1$ and $\mu_2$, the constraint
for $\mu_1 \dpreorder_R \mu_2$ is encoded by means of a weight
function (as given by Definition \ref{def:weight_function_based}) and
the constraint for $\mu_1 \not\dpreorder_R \mu_2$ is encoded by means
of a witness subset of $\supp{\mu_1}$ (as in Lemma
\ref{lem:image_based}), where $R$ is the variable for the strong
simulation.  We use Yices (v$1.0.29$) \cite{yices} to solve
the resulting SMT problem; a
{\em real} variable in Yices input language is essentially a
rational variable.  There is no direct
way to obtain a tree counterexample when the SMT problem is
unsatisfiable. Therefore when the conformance fails, we obtain the
{\em unsat core} from Yices, construct the {\em sub-structure} of
$L_1$ (when we check $L_1 \preorder L_2$) from the constraints in the
unsat core and check the conformance of this sub-structure against
$L_2$ using the \java implementation.  This sub-structure is usually
much smaller than $L_1$ and contains only the information necessary to
expose the counterexample.
%, we expect a significant speed-up when compared to checking the
%conformance entirely with the \java implementation.

%for an efficient implementation of the simulation conformance check which we
%reduced to an SMT problem using linear arithmetic and when a check
%fails, we obtain an LPTS for the {\em unsat core} which is then
%analyzed by our Java code to obtain a tree counterexample. As the
%unsat core is usually much smaller than the original model, this
%speeds-up the whole implementation.

\vspace{0.1in}
\noindent{\bf Results.}
We evaluated our algorithms using this implementation on several
examples analyzed in previous work \cite{FKP_fase11}.
%together with a new system ({\em Ramp}).
Some of these examples were created
by introducing probabilistic failures into non-probabilistic models
used earlier \cite{PGB+_fmsd08} while others were adapted from PRISM
benchmarks \cite{KNP_cav11}. The properties used previously
were about {\em probabilistic reachability} and we had to create our
own specification LPTSes after developing an understanding of the
models. % and the reachability properties.
The models in all the examples satisfy
the respective specifications. We briefly describe the models and the
specifications below, all of which are available at
{\small \url{http://www.cs.cmu.edu/~akomurav/publications/agar/AGAR.html}}.

\begin{description}
\item{{\em CS$_1$ and CS$_N$}} model a {\em Client-Server}
  %protocol with the mutual exclusion property having probabilistic
  protocol with mutual exclusion having probabilistic
  failures in one or all of the $N$ clients, respectively. The specifications describe the probabilistic failure behavior of the clients while hiding some of
the actions as is typical in a high level design specification.

\item{{\em MER}} models a {\em resource arbiter} module of NASA's
  software for {\em Mars Exploration Rovers} which grants and rescinds
  shared resources for several users. We considered the case of two
  resources with varying number of users and probabilistic failures
  introduced in all the components. As in the above example, the specifications describe the
probabilistic failure behavior of the users while hiding some of the actions.

\item{{\em SN}} models a wireless {\em Sensor Network} of one or more
  sensors sending data and messages to a process via a channel with a
  bounded buffer having probabilistic behavior in the components. Creating
specification LPTSes for this example turned out to be more difficult
than the above examples, and we obtained them by observing the system's runs and by manual abstraction.

%\item{{\em Ramp}} models a ramp controller for military operations,
%  and it consists of a ramp battery, circuit breaker and electrical
%  actuater on one side, and the ramp and the controller on the other
%  side. We analyzed the property that the electrical actuator shall
%  not pull more current than a certain threshhold.
\end{description}

\begin{table*}[t]
\scriptsize
\centering
\begin{tabular}{|c|r|r|r|r|r|r|r|r|r|r|r|r|r|r|r|}
\hline
{\em Example} & & & \multicolumn{6}{c|}{\sc ASym} &
\multicolumn{5}{c|}{\sc ASym-N} & \multicolumn{2}{c|}{\sc Mono} \\
\cline{4-16}
$(${\em param}$)$ & $|L|$ & $|P|$ & $|L_1|$ &
$|L_2|$ & {\em Time} & {\em Mem} & $|L_M|$ & $|A_M|$ & $|L_c|$ & {\em Time} & {\em Mem} &
$|L_M|$ & $|A_M|$ & {\em Time} & {\em Mem}\\
\hline

%Client Server 1 failure - Server/Clients
{\em CS$_1 (5)$} & ${\bf 94}$ & $16$ & $36$ & $405$ & $7.2$ & $15.6$ & $182$ &
$33$ & $36$ & $74.0$ & $15.1$ & $182$ & $34$ & ${\bf 0.2}$ & ${\bf 8.8}$\\
{\em CS$_1 (6)$} & ${\bf 136}$ & $19$ & $49$ & $1215$ & $11.6$ & $22.7$ & $324$ &
$41$ & $49$ & $810.7$ & $21.4$ & $324$ & $40$ & ${\bf 0.5}$ & ${\bf 12.2}$\\
{\em CS$_1 (7)$} & ${\bf 186}$ & $22$ & $64$ & $3645$ & $37.7$ & $49.4$ & $538$ &
$56$ & $64$ & {\em out} & -- & -- & -- & ${\bf 0.8}$ & ${\bf 17.9}$\\

\hline

%Client Server N failures - Clients/Server
{\em CS$_N (2)$} & ${\bf 34}$ & $15$ & $25$ & $9$ & $0.7$ & $7.1$ & $51$ &
$7$ & $9$ & $2.4$ & $6.8$ & $40$ & $25$ & ${\bf 0.1}$ & ${\bf 5.9}$\\
{\em CS$_N (3)$} & ${\bf 184}$ & $54$ & $125$ & $16$ & $43.0$ & $63.0$ & $324$ &
$12$ & $16$ & $1.6k$ & $109.6$ & $372$ & $125$ & ${\bf 14.8}$ & ${\bf 37.9}$\\
{\em CS$_N (4)$} & ${\bf 960}$ & $189$ & $625$ & $25$ & {\em out} & -- & -- &
-- & $25$ & {\em out} & -- & -- & -- & ${\bf 1.8}k$ & ${\bf 667.5}$\\

\hline

%MER - Arbiter/Users
{\em MER $(3)$} & $16k$ & $12$ & $278$ & $1728$ & ${\bf 2.6}$ & $19.7$ & ${\bf
706}$ & $7$ & $278$ & $3.6$ & ${\bf 14.6}$ & ${\bf 706}$ & $7$ & $193.8$ & $458.5$\\
{\em MER $(4)$} & $120k$ & $15$ & $465$ & $21k$ & ${\bf 15.0}$ & $53.9$ &
${\bf 2}k$ & $11$ & $465$ & $34.7$ & ${\bf 37.8}$ & ${\bf 2}k$ & $11$ & {\em
out} & --\\
{\em MER $(5)$} & $841k$ & $18$ & $700$ & $250k$ & -- & {\em out}\footnotemark[1] &
-- & -- & $700$ & ${\bf 257.8}$ & ${\bf 65.5}$
& ${\bf 3.3}k$ & $16$ & -- & {\em out}\footnotemark[1]\\

\hline

%Sensor Network - Non-proc/Proc
{\em SN $(1)$} & $462$ & $18$ & $43$ & $32$ & ${\bf 0.2}$ & ${\bf 6.2}$ &
${\bf 43}$ & $3$ & $126$ & $1.7$ & $8.5$ & $165$ & $6$ & $1.5$ & $27.7$\\
{\em SN $(2)$} & $7860$ & $54$ & $796$ & $32$ & ${\bf 79.5}$ & ${\bf
112.9}$ & ${\bf 796}$ & $3$ & $252$ & $694.4$ & $171.7$ & $1.4k$ &
$21$ & $4.7k$ & $1.3k$\\
{\em SN $(3)$} & $78k$ & $162$ & $7545$ & $32$ & {\em out} & -- & -- & -- &
$378$ & ${\bf 7.2}k$ & ${\bf 528.8}$ & ${\bf 1.4}k$ & $21$ & -- & {\em out}\\

\hline
%{\em Ramp } & $127k$ & $1$ & $72$ & $34848$ & $14.9$ &  &  &  &
%& & & & & \\
% add ramp

\end{tabular}
\caption{AGAR vs monolithic verification. $^1$ Mem-out during model construction.}
\label{tab:results}
\vspace{-0.35in}
\end{table*}

Table \ref{tab:results} shows the results we obtained when {\sc ASym}
and {\sc ASym-N} were compared with monolithic (non-compositional)
conformance checking.  $|X|$ stands for the number of states of an
LPTS $X$. $L$ stands for the whole system, $P$ for the specification, $L_M$ for
the LPTS with the largest number of states built by composing LPTSes during the
course of AGAR, $A_M$ for the assumption with the largest
number of states during the execution and $L_c$ for the component with the
largest number of states in {\sc ASym-N}. {\em Time} is in
seconds and {\em Memory} is in megabytes. We also compared $|L_M|$ with $|L|$,
as $|L_M|$ denotes the largest LPTS ever built by AGAR. Best figures, among {\sc
ASym}, {\sc ASym-N} and {\sc Mono}, for {\em Time}, {\em Memory} and LPTS sizes, are boldfaced.
All the results were taken on a Fedora-10 64-bit machine running on an
Intel\textregistered ~Core$^\text{TM}$2 Quad CPU of $2.83$GHz and
$4$GB RAM. We imposed a $2$GB upper bound on Java heap memory and a
$2$ hour upper bound on the running time. We observed that most of the
time during AGAR was spent in checking the
premises and an insignificant amount was spent for the composition and
the refinement steps. Also, most of the memory was consumed by
Yices. 
%We do not report the memory consumed by the explicit-state
%models PRISM constructs.  Further, 
We tried several orderings of the
components (the $L_i$'s in the rules) and report only the ones giving
the best results.

While monolithic checking outperformed AGAR for {\em
  Client-Server}, there are significant time and memory
savings for {\em MER} and {\em Sensor Network} where in
some cases the monolithic approach ran out of resources (time or
memory). One possible reason for AGAR performing worse for {\em
Client-Server} is that $|L|$ is much smaller than $|L_1|$ or $|L_2|$. 
When compared to using {\sc ASym}, \mbox{{\sc ASym-N}}
brings further memory savings in the case of {\em MER} and also time
savings for {\em Sensor Network} with parameter $3$ which could not
finish in $2$ hours when used with {\sc ASym}.
As already mentioned, these models were analyzed previously with an 
assume-guarantee framework using learning from traces \cite{FKP_fase11}.
%As mentioned, a different automated approach  for assume-guarantee reasoning of
%probabilistic systems is presented in \cite{FKP_fase11}, which is also the source for our models.
Although that approach uses a similar assume-guarantee rule (but instantiated to check {\em probabilistic reachability}) and the results have some similarity (e.g. {\em Client-Server} is similarly not handled well by the compositional approach), we 
can not directly compare it with AGAR as it considers
a different class of properties.

\section{Conclusion and Future Work}
%\vspace{-0.1in}
We described a complete, fully automated abstraction-refinement approach for 
assume-guarantee checking of strong simulation between LPTSes.
The approach uses refinement based on counterexamples formalized as stochastic
trees and it further applies to checking {\em safe}-pCTL properties. 
We showed experimentally the merits of the proposed technique. 
We plan to extend our approach to cases where the assumption $A$ has a smaller
alphabet than that of the component it represents as this can
potentially lead to further savings. Strong
simulation would no longer work and one would need to use
{\em weak} simulation \cite{SL_nordic95}, for which checking algorithms are
unknown yet.
%We further remark that, although it uses an explicit representation
%for individual components (e.g. $L_2$), our approach never builds $L_1 \parallel
%L_2$ directly, keeping the cost of verification small. We would like, however,
We would also like
to explore symbolic implementations of our algorithms, for increased scalability. %This would %involve devolping new CEGAR and simulation checking approaches for a symbolic representation of LPTSes. 
As an alternative approach, we plan to build upon our recent work \cite{KPC_lics12} on
learning LPTSes to develop practical compositional algorithms and compare with AGAR.

%Finally, we note that in recent work \cite{KPC_lics12}, we investigated algorithms for learning
%unknown non-deterministic LPTSes from tree samples and 
%described how to apply them to learn assumptions for compositional
%verification. The result is a doubly-exponential algorithm in
%the size of $L_2$. In the future, we would like to explore more efficient learning algorithms and 
%compare them with AGAR.
%, although similar studies in the non-probabilistic case where inconclusive \cite{BPG_cav08,BHS_cav07}.

%Finally we would like to experimentally compare AGAR with related approaches, e.g. \cite{} and to investigate %how often incompleteness is really an issue in practice for those approaches.
%{\bf mention we would like to look at circular rules (Corina)}

%{\bf also mention we plan to look into symbolic algorithms}

%{\bf mention somewhere our lics submission}

\subsubsection*{Acknowledgments.}
We thank
Christel Baier,
Rohit Chadha,
Lu Feng,
Holger Hermanns,
Marta Kwiatkowska,
Joel Ouaknine,
David Parker,
Frits Vaandrager,
Mahesh Viswanathan,
James Worrell and
Lijun Zhang
%C. Baier,
%R. Chadha,
%L. Feng,
%H. Hermanns,
%M. Kwiatkowska,
%J. Ouaknine,
%D. Parker,
%F. Vaandrager,
%M. Viswanathan,
%J. Worrell and
%L. Zhang
for generously answering our questions related to
this research. We also
thank the anonymous reviewers for their suggestions and
David Henriques
%D. Henriques
for
carefully reading an earlier draft.

\newpage
\appendix

\section{Proof of Lemma \ref{lem:precongruence}}
We first show that $\preorder$ is a preorder. Reflexivity can be easily proved by
showing that the identity relation is a strong simulation. We only consider
transitivity. Let $L_1 \preorder L_2$ and $L_2 \preorder L_3$. Thus, there are
strong simulations $R_{12} \subseteq S_1 \times S_2$ and $R_{23} \subseteq S_2
\times S_3$. Consider the relation $R = \{(s_1,s_3) | \exists s_2 : s_1 R_{12}
s_2 ~\text{and}~ s_2 R_{23} s_3\}$. Let $s_1 R s_3$ and $s_1 \trans{a}
\mu_1$. Also, let $s_2 \in S_2$ be such that $s_1 R_{12} s_2$ and $s_2 R_{23}
s_3$. As $R_{12}$ is a strong simulation, there exists $s_2 \trans{a} \mu_2$
with $\mu_1 \dpreorder_{R_{12}} \mu_2$. Again, as $R_{23}$ is a strong
simulation, there exists $s_3 \trans{a} \mu_3$ with $\mu_2
\dpreorder_{R_{23}} \mu_3$. Now, let $S \subseteq \supp{\mu_1}$ be
arbitrary. We have $\mu_1(S) \leq \mu_2(R_{12}(S)) \leq \mu_3(R_{23}(R_{12}(S))) =
\mu_3(R(S))$ (Lemma \ref{lem:image_based}). Thus, $\mu_1 \dpreorder_R \mu_3$ and hence, $R$ is a
strong simulation. Also, $s^0_1 R s^0_3$ by definition of $R$. We conclude that $L_1 \preorder L_3$.

Now, we show that $\preorder$ is compositional. Assume $L_1 \preorder L_2$ with
$\alpha_2 \subseteq \alpha_1$. Let $R_{12} \subseteq S_1 \times S_2$ be a strong
simulation. Consider the relation $R$ defined below.
\[
R = \{((s_1,s),(s_2,s)) | s_1 R_{12} s_2 ~\text{and}~ s \in S_L\}
\]

Let $(s_1,s) R (s_2,s)$ and $(s_1,s) \trans{a}
\mu_a$. So, $s_1 R_{12} s_2$. By Definition \ref{def:composition}, there
are three cases to analyze.

\begin{description}
\item{$s_1 \trans{a} \mu_1$, $s \trans{a} \mu$ and $\mu_a = \mu_1
\product \mu$ :} As $R_{12}$ is a strong simulation, there exists $s_2
\trans{a} \mu_2$ with $\mu_1 \dpreorder_{R_{12}} \mu_2$. And by
Definition \ref{def:composition}, $(s_2,s) \trans{a} \mu'_a$ where $\mu'_a =
\mu_2 \product \mu$. Now, let $X \subseteq \supp{\mu_a}$. For each $s \in
S_L$, let $X_s \subseteq X$ contain all the pairs of $X$ with $s$ as the second
member. Thus, the $X_s$'s partition $X$. We have $\mu_a(X)$
  \begin{align*}
  & = \sum_{s \in S_L} \mu_a(X_s \times \{s\})\\
  & = \sum_{s \in S_L} \mu_1(X_s) \cdot \mu(s) & \text{definition of
$\mu_a$}\\
  & \leq \sum_{s \in S_L} \mu_2(R_{12}(X_s)) \cdot \mu(s) & \text{$R_{12}$
is a strong simulation}\\
  & = \sum_{s \in S_L} \mu'_a(R_{12}(X_s) \times \{s\}) & \text{definition of
$\mu'_a$}\\
  & = \sum_{s \in S_L} \mu'_a(R(X_s \times \{s\})) & \text{definition of
$R$}\\
  & = \mu'_a(\bigcup_{s \in S_L} R(X_s \times \{s\})) & \text{the sets $R(X_s \times
\{s\})$ are disjoint for distinct $s$}\\
  & = \mu'_a(R(\bigcup_{s \in S_L} X_s \times \{s\})) &\\
  & = \mu'_a(R(X))
  \end{align*}
which implies that $\mu_a \dpreorder_R \mu'_a$.

\item{$a \not\in \alpha_1$, $s \trans{a} \mu$ and $\mu_a = \dirac{s_1}
\product \mu$ :} As $\alpha_2 \subseteq \alpha_1$, $a \not\in \alpha_2$ and
by Definition \ref{def:composition}, $(s_2,s) \trans{a} \mu'_a$ with
$\mu'_a = \dirac{s_2} \product \mu$. Now, let $X \subseteq
\supp{\mu_a}$ and let $X_2$ denote the set of all the second members of the
pairs in $X$. We have $\mu_a(X) = \mu(X_2) = \mu'_a(\{s_2\} \times X_2)
\leq \mu'_a(R(X))$ and hence,
$\mu_a \dpreorder_R \mu'_a$.

\item{$s_1 \trans{a} \mu_1$, $a \not\in \alpha_L$ and $\mu_a = \mu_1
\product \dirac{s}$ : } As $R_{12}$ is a strong simulation, there exists $s_2
\trans{a} \mu_2$ with $\mu_1 \dpreorder_{R_{12}} \mu_2$. Now, let $X
\subseteq \supp{\mu_a}$ and let $X_1$ denote the set of all the first members
of the pairs in $X$. We have $\mu_a(X) = \mu_1(X_1) \leq
\mu_2(R_{12}(X_1)) = \mu'_a(R(X))$ and hence, $\mu_a \dpreorder_R
\mu'_a$.
\end{description}

Hence, $R$ is a strong simulation. Also, $(s^0_1,s^0_L) R
(s^0_2,s^0_L)$ by definition of $R$. We conclude that $L_1 \parallel L \preorder
L_2 \parallel L$.
\qed

\section{Proof of Theorem \ref{thm:tree_cex}}

We give a constructive proof. Assume that $L_1 \not\preorder L_2$.

We first describe, briefly, a well-known algorithm used to check $L_1
\preorder L_2$ \cite{Baier_habilitation98}. We start with a candidate $R$ for the coarsest strong simulation
between $L_1$ and $L_2$ initialized to $S_1 \times S_2$. Each iteration, an
arbitrary pair $(s_1,s_2)$ in the current $R$ is picked and the local conditions
in the definition of a strong simulation (Definition
\ref{def:strong_simulation}) are checked for $R$. If the pair
fails, that is because there is a transition $s_1 \trans{a} \mu_1$ but for
every $s_2 \trans{a} \mu_2$, $\mu_1 \not\dpreorder_R
\mu_2$. In this case, the pair is removed and another iteration begins. Note
that, at this point we can conclude that $s_1 \not\preorder s_2$.
Otherwise, a new pair is picked for examination. The algorithm stops when
$(s^0_1,s^0_2)$ (the pair of the initial states) is removed from the current
$R$ at which point we conclude that $L_1 \not\preorder L_2$, or when a fixed
point is reached and we conclude that $L_1 \preorder L_2$. By the correctness
and termination of this algorithm, this will eventually happen. And by the
assumption made above that $L_1 \not\preorder L_2$, we are only interested in
the former scenario of termination.

We show that whenever a pair $(s_1,s_2)$ is removed from $R$, there is a
tree $T_{12}$ which serves as a counterexample to $s_1
\preorder s_2$. As argued above, $(s^0_1,s^0_2)$ is eventually removed
from $R$ and hence, we have a tree $T$ which serves a
counterexample to $s^0_1 \preorder s^0_2$ and therefore, to the conformance.
We proceed by strong induction on the number of pairs
removed so far from the initial $R = S_1 \times S_2$.

%The base case is when no pair has been removed so far. In this case, $(s_1,s_2)$
%will be removed only because there is a transition $s_1 \trans{a} \mu_1$ and for every $s_2 \trans{a}
%\mu_2$ with $\mu_2 \in \dist{S_2}$ (there might be none), $\mu_1
%\not\dpreorder_R \mu_2$, \ie $\mu_1(S_1) > \mu_2(S_2)$ (as $R = S_1
%\times S_2$). In general, we can find a subset $S \subseteq S_1$ such that
%$\mu_1(S) > \mu_2(S_2)$ for each such $\mu_2$.
%Then, a counterexample will simply be the tree $T_{12}$ representing
%the transition $s_1 \trans{a} \mu'_1$ where $\mu'_1$ is the
%distribution obtained when $\mu_1$ is restricted to $S$. It is easy to see that $T_{12} \preorder
%(L_1,s_1)$ but $T_{12} \not\preorder (L_2,s_2)$.

The base case is when no pair has been removed so far. In this case, $(s_1,s_2)$
will be removed only because there is a transition $s_1 \trans{a} \mu_1$ and
there is no transition on action $a$ from $s_2$.
Then, a counterexample will simply be the tree $T_{12}$ representing
the transition $s_1 \trans{a} \mu_1$. It is easy to see that $T_{12} \preorder
(L_1,s_1)$ but $T_{12} \not\preorder (L_2,s_2)$.

For the inductive case, assume that a new pair $(s_1,s_2)$ has been removed from
the current $R$. We have to analyze two cases. The first case is when we have
a transition $s_1 \trans{a} \mu_1$ but there is no transition $s_2 \trans{a} \mu_2$. This
is similar to the base case above. So, we will only consider the other case
below.

Now, there is a transition $s_1 \trans{a} \mu_1$ and the set $\Delta = \{\mu \in
\dist{S_2} | s_2 \trans{a} \mu\}$ is non-empty but for every $\mu \in \Delta$,
$\mu_1 \not\dpreorder_R \mu$. Consider an arbitrary $\mu \in \Delta$.
Because $\mu_1 \not\dpreorder_R \mu$, we conclude that there is a set $S^\mu_1
\subseteq \supp{\mu_1}$ such that $\mu_1(S^\mu_1) > \mu(R(S^\mu_1))$
(Lemma \ref{lem:image_based}). Intuitively, this is because
$S^\mu_1$ is {\em not related to enough number of states} from
$\supp{\mu}$. Let $S^\mu_2 = \supp{\mu} \setminus R(S^\mu_1)$.

%We start building a tree $T_{12}$ with $s_1$ as the root and $s_1
%\trans{a} \mu'_1$ as the only outgoing transition where $\mu'_1$ is the
%distribution obtained when $\mu_1$ is restricted to $\bigcup_{\mu \in
%\Delta} S^\mu_1$. Now, let $s \in \supp{\mu'_1}$. Consider the set $U_s =
%\bigcup \{S^\mu_2 | s \in S^\mu_1\}$. Then, for every $t \in U_s$, we simply
%attach the counterexample tree for $(s,t)$ (exists by
%induction hypothesis) below the state $s$ in
%$T_{12}$. We claim that $T_{12}$ built this way is a counterexample to
%$s_1 \preorder s_2$.

We start building a tree $T_{12}$ with $s_1$ as the root and $s_1
\trans{a} \mu_1$ as the only outgoing transition. Now, let $s \in \bigcup_{\mu
\in \Delta} S^\mu_1$. Consider the set $U_s =
\bigcup \{S^\mu_2 | s \in S^\mu_1\}$. Then, for every $t \in U_s$, we simply
attach the counterexample tree for $(s,t)$ (exists by
induction hypothesis) below the state $s$ in
$T_{12}$. We claim that $T_{12}$ built this way is a counterexample to
$s_1 \preorder s_2$.

First of all, it is easy to see that $T_{12} \preorder (L_1,s_1)$ as $T_{12}$ is
obtained from the states and the corresponding distributions of $L_1$. Let $\mu \in
\Delta$ and let $R'$ be a strong simulation between $T_{12}$ and $L_2$.
By construction, $S^\mu_1 \subseteq \supp{\mu_1}$ and further, by induction
hypothesis for every $(s,t) \in S^\mu_1 \times S^\mu_2$, $(T_{12},s)
\not\preorder (L_2,t)$ and hence, $(s,t) \not\in R'$. Therefore
$\mu_1(S^\mu_1) > \mu(R(S^\mu_1)) \ge
\mu(R'(S^\mu_1))$. It follows that $\mu_1 \not\dpreorder_{R'} \mu$ and hence, $(s_1,s_2)
\not\in R'$. As $\mu$ and $R'$ are arbitrary, we
conclude that $T_{12} \not\preorder (L_2,s_2)$.
\qed

\section{Proof of Theorem \ref{thm:cex_complexity}}
It can be easily be seen that Algorithm \ref{algo:dist_cex} takes $O(n^3)$ time and
$O(n)$ space which increases the complexity of checking $\mu_1 \dpreorder_R
\mu_2$ to $O(n^3)$ time and $O(n^2)$ space (see Section \ref{sec:prelims}).
The rest of the argument is similar to that of the fixed point algorithm for computing
the coarsest strong simulation \cite{Baier_habilitation98}.
\qed

\section {Proof of Lemma \ref{lem:cex_fully_prob_insuff}}

\begin{figure}
\centering
\includegraphics[scale=1.2]{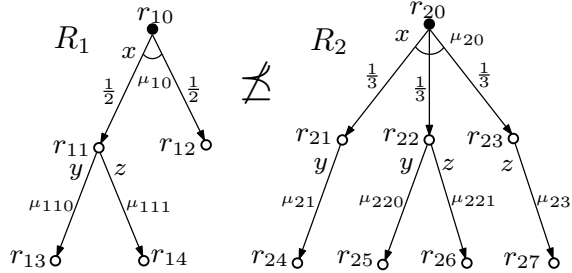}
\caption{An example where there is no fully-probabilistic counterexample.}
\label{fig:cex_fully_prob_insuff}
\end{figure}

Consider the two reactive LPTSes $R_1$ and $R_2$ in Figure
\ref{fig:cex_fully_prob_insuff}. The
states along with the outgoing actions and distributions are labeled as in the figure.
Clearly $r_{11} \not\preorder r_{21}$ and $r_{11} \not\preorder r_{23}$. It follows
that $\mu_{10} \not\dpreorder_\preorder \mu_{20}$ and hence, $R_1 \not\preorder R_2$. We are
interested in a counterexample to demonstrate this.

Let us assume that there is a fully-probabilistic LPTS $C$ (with initial state $c_0$) which
serves as a counterexample. Thus, $C \preorder R_1$ but $C \not\preorder R_2$. By Definition
\ref{def:strong_simulation} there exists a strong simulation $U$ such that $(c_0, r_{10}) \in
U$. If $c_0$ has no outgoing transitions, clearly $C \preorder R_2$. So,
it must have an outgoing distribution, say $\mu_0$. As $(c_0,r_{10}) \in
U$ and as $\mu_{10}$ is labeled by $x$, $\mu_0$ must be labeled by
$x$ too. Let $c_1$ be an arbitrary state in $\supp{\mu_0}$ with an outgoing
transition (there may be no such $c_1$). Then, the transition must be labeled by $y$ or $z$. Otherwise, clearly
$(c_1,r_{11}) \not\in U$ and $(c_1,r_{12}) \not\in U$ which imply $\mu_0
\not\dpreorder_{U} \mu_{10}$ and hence, $(c_0,r_{10}) \not\in U$
contradicting the assumption. Moreover, $(c_1,r_{12}) \not\in U$ as $r_{12}$
has no transitions. This forces $(c_1,r_{11})$ to be in $U$. Let the (only) outgoing
distribution $\mu_1$ of $c_1$ be labeled by $y$. Then, for every state $c_2 \in \supp{\mu_1}$,
$(c_2,r_{13}) \in U$ for otherwise $\mu_1 \not\dpreorder_{U}
\mu_{110}$ which implies $(c_1,r_{11}) \not\in U$ leading to a contradiction.
This forces $c_2$ to not have any transitions. We have the same
conclusion if $\mu_1$ is labeled by $z$ instead.

Thus, $C$ can only be a tree with exactly one transition
$\mu_0$ labeled by $x$ from the initial state and for every state in the
support of this distribution, there is at most one transition labeled by
either $y$ or $z$. Also, if $S_y$ and $S_z$ are the sets of states in
$\supp{\mu_0}$ with a transition labeled by $y$ and $z$,
respectively, then $\mu_0(S_y \cup S_z) \le \frac{1}{2}$. This is because, $U(S_y
\cup S_z) = \{r_{11}\}$ and $\mu_{10}(r_{11}) = \frac{1}{2}$.
%It follows that $\mu_0(S_y) \le \frac{1}{2}$ and $\mu_0(S_z) \le \frac{1}{2}$.

Now, we define a relation $V$ between the states of $C$, $S_C$, and that of
$R_2$, $S_2$. The initial states are related. Let $c$ be an arbitrary state of
$C$. If $c$ has no transitions it is related to every state of $R_2$. If $c$ has
its transition labeled by $y$, it is related to $r_{21}$ and $r_{22}$.
Otherwise its transition is labeled by $z$ and it is related to
$r_{22}$ and $r_{23}$. To show that $V$ is a strong simulation, the
only non-trivial thing to consider is whether $\mu_0 \dpreorder_{V}
\mu_{20}$. For that, take an arbitrary set $X \subseteq \supp{\mu_0}$. If $X$ has any state
with no transitions, $V(X) = S_2$ and hence $\mu_0(X) \le
\mu_{20}(V(X)) = 1$. Otherwise, $X$ only has states with transitions
labeled by $y$ or $z$, \ie $X \subseteq S_y \cup S_z$, and by the observation made in the above paragraph,
$\mu_0(X) \le \frac{1}{2}$ whereas $\mu_{20}(V(X)) \ge \frac{2}{3}$. Thus,
$\mu_0(X) \le \mu_{20}(V(X))$. This shows that $V$ is a strong
simulation and we conclude that $C \preorder R_2$ immediately giving us a contradiction to the
assumption that $C$ is a counterexample.
\qed

\section{Proof of Lemma \ref{lem:cex_reactive_insuff}}

\begin{figure}
\centering
\includegraphics[scale=1.2]{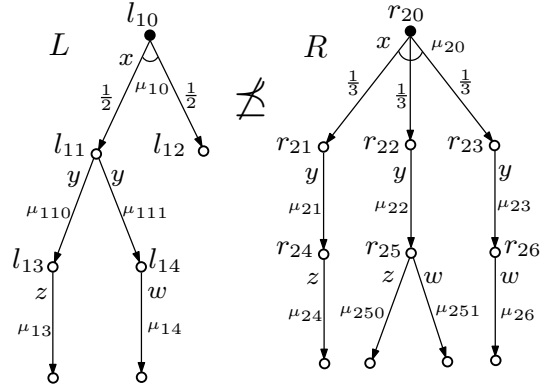}
\caption{There is no reactive counterexample to $L \preorder R$.}
\label{fig:cex_reactive_insuff}
\end{figure}

Consider the LPTS $L$ and the reactive LPTS $R$ in Figure \ref{fig:cex_reactive_insuff}. The
states along with the outgoing actions and distributions are labeled as in the
figure. By similar arguments as made in the proof of Lemma
\ref{lem:cex_fully_prob_insuff}, one can show that
$\mu_{110} \not\dpreorder_\preorder \mu_{23}$, $\mu_{111}
\not\dpreorder_\preorder \mu_{21}$ whereas $\mu_{110} \dpreorder_\preorder
\mu_{21}, \mu_{22}$ and $\mu_{111} \dpreorder_\preorder
\mu_{22}, \mu_{23}$. All these imply that $L \not\preorder R$. We are
interested in a counterexample to show this.

Assume that a reactive LPTS $C$ exists which serves as a counterexample. Again, similar
to the arguments made in the proof of Lemma \ref{lem:cex_fully_prob_insuff}, one can
show that $C$ can only be a tree with exactly one transition
$\mu_0$ labeled by $x$ from the initial state and for
every state in $\supp{\mu_0}$, there is at most one distribution labeled by
$y$ (because $l_{11}$ has transitions on no other action). Furthermore, if any state in the support of this distribution has any
transitions, all the transitions from all the states in the support will be
labeled by the same action and that too, by either $z$ or $w$.  
Then, if $S_y$ is the set of states in $\supp{\mu_0}$ with outgoing
distributions (which should only be labeled by $y$) then $\mu_0(S_y) \le \frac{1}{2}$.

Now, we define a relation $V \subseteq S_C \times S_2$, where $S_C$ is the set
of states of $C$, in a similar fashion. All the states in $C$ with no transitions
are related to every state in $S_2$. The initial states are related. For every other state
$c$, if it has a transition labeled by $z$ or $w$, $c$ is related
to all the states having a transition on $z$ or $w$, respectively and if it is
labeled by $y$, it is related to $r_{21}$ ($r_{23}$) and $r_{22}$ if the states in the
support have transitions on $z$ ($w$) and to all three of $r_{21}$, $r_{22}$ and
$r_{23}$ otherwise. One can similarly show that $V$ is a strong simulation
implying $C \preorder R$. This contradicts the assumption that $C$ is a
counterexample.
\qed

\section{Quotient is an Abstraction : $L \preorder L/\Pi$}
It suffices to show that $R = \{(s,c) | s \in c, c \in \Pi\}$ is a strong
simulation between $L$ and $L/\Pi$. Let $s R c$ and $s \trans{a} \mu$. As $s \in
c$, there exists a transition, by Definition~\ref{def:quotient_lpts}, $c
\trans{a} \mu_l$ such that for every $c' \in \Pi$, $\mu_l(c') = \sum_{s' \in c'}
\mu(s')$. Let $S \subseteq S_L$. Now, $\mu(S)$

  \begin{align*}
  & = \sum_{s' \in S} \mu(s') \\
  & = \sum_{c' \in R(S)} \sum_{s' \in c' \cap S} \mu(s') \\
  & \le \sum_{c' \in R(S)} \sum_{s' \in c'} \mu(s') \\
  & = \sum_{c' \in R(S)} \mu_l(c') \\
  & = \mu_l(R(S))
  \end{align*}

As $S$ is arbitrary, this implies from Lemma~\ref{lem:image_based} that $\mu
\dpreorder_R \mu_l$. Note that $s^0_L R [s^0_L]$.

We conclude that $L \preorder L/\Pi$.
\qed

\section{Proof of Lemma \ref{lem:ref_progress}}

Let $s_1$, $\mu_1$ and $M$ be as in Section \ref{sec:cegar}. Consider the first case
where $R(s_1) = \emptyset$. If $R_\old(s_1) = R_M(s_1)$, it follows that there exists $s \in
\supp{\mu_1}$ with $R_\old(s) \subset R_M(s)$. This can be easily proved by
contradiction and we omit this proof. As $M(s)$ is split into $R_\old(s)$ and
the rest, the strategy results in a finer partition. Otherwise,
$R_\old(s_1)$ is a strict subset of $R_M(s_1)$ and as $R(s_1) = \emptyset$, the
strategy splits $M(s_1)$ into $R_\old(s_1)$ and the rest which also results in
a finer partition.

Now, consider the second case where $R(s_1) \neq \emptyset$, $M(s_1) =
[s^0_L]_\Pi$ and $s^0_L \in R_\old(s_1) \setminus R(s_1)$. It follows that $R_\old(s_1)
\setminus R(s_1)$ is a non-empty, proper subset of $R_M(s_1)$ and hence, this
also results in a finer partition.
\qed

%\section{Proof of Lemma \ref{lem:stronger_spuriousness_check}}
%Assume the contrary that $C \project{A} \preorder L_2$. So, $(C
%\project{A})^{\alpha_2} \preorder L_2$. We have seen that $(C
%\project{L_1})^{\alpha_1} \parallel (C \project{A})^{\alpha_2} \not\preorder P$.
%But Lemma \ref{lem:precongruence} gives us $(C \project{L_1})^{\alpha_1}
%\parallel (C \project{A})^{\alpha_2} \preorder (C \project{L_1})^{\alpha_1} \parallel L_2 \preorder P$ and
%by transitivity of $\preorder$, we obtain a contradiction.
%\qed

\end{document}